\documentclass[submission,copyright,creativecommons]{eptcs}
\providecommand{\event}{Festschrift Symposium for Dave Schmidt} 

\usepackage{amssymb,stmaryrd,graphics,color}
\usepackage{ifpdf,graphicx,fancyvrb}
\usepackage[all]{xy}

\newcommand{\li}{\ar@{-}} 
\newcommand{\lp}{\ar@{.}} 
\newcommand{\fp}{\ar@{.>}} 

\newcommand{\declassify}{\mbox{declassify}}

\newcounter{todos}
\def\post{\mbox{\sl post}}
\def\pre{\mbox{\sl pre}}

\newcommand{\nint}{\mathbb{Z}}

\newcommand{\val}{\mathbb{V}}
\newcommand{\typ}{\mathbb{T}}

\newcommand{\modu}{\ \mbox{\sl mod}\ }

\newcommand{\uco}{\mbox{\sl uco}}

\newcommand{\id}{\mbox{\sl id}}

\newcommand{\Var}{\mbox{\sl Var}}
\newcommand{\ra}{\rightarrow}
\newcommand{\la}{\leftarrow}
\newcommand{\Ra}{\Rightarrow}

\newcommand{\Lra}{\Leftrightarrow}

\newcommand{\undist}[2]{\Upsilon_{\mbox{\tiny $#1$}}^{\mbox{\tiny $#2$}}}
\newcommand{\Secr}[2]{\mbox{\sl Secr}_{\mbox{\tiny $#1$}}^{#2}}
\def\defi{\mbox{\raisebox{0ex}[1ex][1ex]{$\stackrel{\mbox{\tiny
def}}{\; =\;}$}}}

\def\ok#1{\mbox{\raisebox{0ex}[1ex][1ex]{$#1$}}}

\newcommand{\comment} [1]{}

\def\Sign{\mbox{\sl Sign\/}}

\def\Par{\mbox{\sl Par\/}}

\def\comp{\mathrel{\hbox{\footnotesize${}\!{\circ}\!{}$\normalsize}}}

\def\grasstr#1{\langle\!|#1|\!\rangle}
\def\grass#1{\llbracket#1\rrbracket}

\def\defemb#1#2{\expandafter\def\csname #1\endcsname
                              {\relax\ifmmode #2\else\hbox{$#2$}\fi}}
 
\def\2c-math#1#2{{\par\medskip\noindent ${#1}$
                      \par\smallskip
                        \noindent\hspace*{\fill} ${#2}$}
                           \\[10pt]}

\def\cal{\mathcal}			   
\defemb{aF}{{\cal F}^\cA}
\defemb{sF}{{\cal F}^{\star}}
\defemb{cA}{{\cal A}}
\defemb{cB}{{\cal B}}
\defemb{cC}{{\cal C}}
\defemb{cD}{{\cal D}}
\defemb{cE}{{\cal E}}
\defemb{cF}{{\cal F}}
\defemb{cG}{{\cal G}}
\defemb{cH}{{\cal H}}
\defemb{cI}{{\cal I}}
\defemb{cJ}{{\cal J}}
\defemb{cK}{{\cal K}}
\defemb{cL}{{\cal L}}
\defemb{cM}{{\cal M}}
\defemb{cN}{{\cal N}}
\defemb{cO}{{\cal O}}
\defemb{cP}{{\cal P}}
\defemb{cQ}{{\cal Q}}
\defemb{cR}{{\cal R}}
\defemb{cS}{{\cal S}}
\defemb{cT}{{\cal T}}
\defemb{cU}{{\cal U}}
\defemb{cV}{{\cal V}}
\defemb{cW}{{\cal W}}
\defemb{cY}{{\cal Y}}
\defemb{cX}{{\cal X}}
\defemb{cZ}{{\cal Z}}
\defemb{cHys}{{\cH_{\cY}^{\cS}}}
\defemb{cHy}{{\cH_{\cY}}}
\defemb{ld}{\dashv}
\defemb{rd}{\vdash}
\defemb{cGH}{{\cal GH}}
\defemb{cGWP}{{\cal GWP}}
\defemb{tL}{{\tt L}}
\defemb{tH}{{\tt H}}
\defemb{tT}{{\tt T}}
\defemb{tS}{{\tt S}}
\defemb{tP}{{\tt P}}
\defemb{tQ}{{\tt Q}}
\defemb{tR}{{\tt R}}
\defemb{tB}{{\tt B\:}}
\defemb{tA}{{\tt A\:}}
\defemb{tC}{{\tt C\:}}

\newcommand{\sset}[2]{\left\{~#1  \left |
                               \begin{array}{l}#2\end{array}
                          \right.     \right\}}

\newcommand{\dsecr}[3]{\mbox{\small $[#1]$}#2\mbox{\small $(#3)$}}
\newcommand{\asecr}[3]{\mbox{\small $(#1)$}#2\mbox{\small $(#3)$}}
\newcommand{\gsecr}[4]{\mbox{\small $(#1)$}#2\mbox{\small $(#3\rightsquigarrow\!\!\![\!] #4)$}}
\newcommand{\gnsecr}[4]{\mbox{\small $(#1)$}#2\mbox{\small $(#3\Ra #4)$}}

\newcommand{\ifc}{\mbox{\bf if}}
\newcommand{\thenc}{\mbox{\bf then}}
\newcommand{\elsec}{\mbox{\bf else}}
\newcommand{\nil}{\mbox{\bf skip}}

\newcommand{\while}{\mbox{\bf while}}

\newcommand{\dow}{\mbox{\bf do}}
\newcommand{\ew}{\mbox{\bf endw}}

\def\tuple#1{\langle #1 \rangle}

\newcommand{\fP}{\grass{P}} 

\newcommand{\ANI}{\mbox{\tiny{\sl ANI}}}
\newcommand{\NANI}{\mbox{\tt NANI}}

\defemb{tD}{{\tt D\:}}

\defemb{tX}{{\tt X}}

\newcommand{\Nats}{\mathsf{Nats}}

\newcommand{\Wlp}{\mbox{\sl Wlp}}

\newcommand{\f}{\ar@{->}} 

\DeclareSymbolFont{italics}{OT1}{cmr}{m}{it}

\DeclareMathSymbol{a}{\mathalpha}{italics}{"61}
\DeclareMathSymbol{b}{\mathalpha}{italics}{"62}
\DeclareMathSymbol{c}{\mathalpha}{italics}{"63}
\DeclareMathSymbol{d}{\mathalpha}{italics}{"64}
\DeclareMathSymbol{e}{\mathalpha}{italics}{"65}
\DeclareMathSymbol{f}{\mathalpha}{italics}{"66}
\DeclareMathSymbol{g}{\mathalpha}{italics}{"67}
\DeclareMathSymbol{h}{\mathalpha}{italics}{"68}
\DeclareMathSymbol{i}{\mathalpha}{italics}{"69}
\DeclareMathSymbol{j}{\mathalpha}{italics}{"6A}
\DeclareMathSymbol{k}{\mathalpha}{italics}{"6B}
\DeclareMathSymbol{l}{\mathalpha}{italics}{"6C}
\DeclareMathSymbol{m}{\mathalpha}{italics}{"6D}
\DeclareMathSymbol{n}{\mathalpha}{italics}{"6E}
\DeclareMathSymbol{o}{\mathalpha}{italics}{"6F}
\DeclareMathSymbol{p}{\mathalpha}{italics}{"70}
\DeclareMathSymbol{q}{\mathalpha}{italics}{"71}
\DeclareMathSymbol{r}{\mathalpha}{italics}{"72}
\DeclareMathSymbol{s}{\mathalpha}{italics}{"73}
\DeclareMathSymbol{t}{\mathalpha}{italics}{"74}
\DeclareMathSymbol{u}{\mathalpha}{italics}{"75}
\DeclareMathSymbol{v}{\mathalpha}{italics}{"76}
\DeclareMathSymbol{w}{\mathalpha}{italics}{"77}
\DeclareMathSymbol{x}{\mathalpha}{italics}{"78}
\DeclareMathSymbol{y}{\mathalpha}{italics}{"79}
\DeclareMathSymbol{z}{\mathalpha}{italics}{"7A}

\DeclareMathSymbol{A}{\mathalpha}{italics}{"41}
\DeclareMathSymbol{B}{\mathalpha}{italics}{"42}
\DeclareMathSymbol{C}{\mathalpha}{italics}{"43}
\DeclareMathSymbol{D}{\mathalpha}{italics}{"44}
\DeclareMathSymbol{E}{\mathalpha}{italics}{"45}
\DeclareMathSymbol{F}{\mathalpha}{italics}{"46}
\DeclareMathSymbol{G}{\mathalpha}{italics}{"47}
\DeclareMathSymbol{H}{\mathalpha}{italics}{"48}
\DeclareMathSymbol{I}{\mathalpha}{italics}{"49}
\DeclareMathSymbol{J}{\mathalpha}{italics}{"4A}
\DeclareMathSymbol{K}{\mathalpha}{italics}{"4B}
\DeclareMathSymbol{L}{\mathalpha}{italics}{"4C}
\DeclareMathSymbol{M}{\mathalpha}{italics}{"4D}
\DeclareMathSymbol{N}{\mathalpha}{italics}{"4E}
\DeclareMathSymbol{O}{\mathalpha}{italics}{"4F}
\DeclareMathSymbol{P}{\mathalpha}{italics}{"50}
\DeclareMathSymbol{Q}{\mathalpha}{italics}{"51}
\DeclareMathSymbol{R}{\mathalpha}{italics}{"52}
\DeclareMathSymbol{S}{\mathalpha}{italics}{"53}
\DeclareMathSymbol{T}{\mathalpha}{italics}{"54}
\DeclareMathSymbol{U}{\mathalpha}{italics}{"55}
\DeclareMathSymbol{V}{\mathalpha}{italics}{"56}
\DeclareMathSymbol{W}{\mathalpha}{italics}{"57}
\DeclareMathSymbol{X}{\mathalpha}{italics}{"58}
\DeclareMathSymbol{Y}{\mathalpha}{italics}{"59}
\DeclareMathSymbol{Z}{\mathalpha}{italics}{"5A}

\newcommand{\even}{\mbox{\tt ev}}
\newcommand{\odd}{\mbox{\tt od}}

\newcommand{\rel}{*}
\newcommand{\obs}{\circ}

\newcommand{\Pinp}{\cI}
\newcommand{\Poup}{\cO}
\newcommand{\mtI}{\mathbb{I}}
\newcommand{\mtO}{\mathbb{O}}
\newcommand{\inp}{\mtI}
\newcommand{\oup}{\mtO}
\newcommand{\Prelin}{\Pinp_{\rel}}
\newcommand{\Prelout}{\Poup_{\rel}}
\newcommand{\Pobsin}{\Pinp_{\obs}}
\newcommand{\Pobsout}{\Poup_{\obs}}
\newcommand{\relin}{\inp_{\rel}}
\newcommand{\relout}{\oup_{\rel}}
\newcommand{\obsin}{\inp_{\obs}}
\newcommand{\obsout}{\oup_{\obs}}
\def\mtC{{\mathbb{C}}}
\newcommand{\Oout}{\rho}
\newcommand{\Oin}{\eta}
\newcommand{\Sel}{\phi}
\newcommand{\gani}[4]{[#1\vdash(#2)#3(#4)]}

\newtheorem{theorem}{Theorem}[section]
\newtheorem{proposition}[theorem]{Proposition}
\newtheorem{example}[theorem]{Example}

\newtheorem{definition}[theorem]{Definition}

\newenvironment{proof}{\noindent {\sc Proof.~}}{\hfill 
$\Box$\newline\smallskip\mbox{}\unskip~~\normalsize} 

\newcommand{\topi}{\typ^{\rel}}

\newcommand{\ido}{\id^{\obs}}

\newcommand{\ev}{\mbox{\tt ev}}

\newcommand{\tM}{\mathbb{M}}
\newcommand{\prog}[1]{\ttP_{#1}}
\newcommand{\ttP}{\mathbb{P}}

\newcommand{\tpost}{\rightarrow_{P}}
\defemb{tO}{{\tt O}}
\defemb{tI}{{\tt I}}

\newcommand{\dgreen}{\color[rgb]{0,.55,0}}
\newcommand{\dblue}{\color[rgb]{0,0,.6}}
\newcommand{\dred}{\color[rgb]{.7,0,0}}

\def\NI{\mbox{\tt NI}}
\def\ANI{\mbox{\tt ANI}}
\def\DNI{\mbox{\tt DNI}}

\def\DANI{\mbox{\tt DANI}}

\def\BDNI{\mbox{\tt B-DNI}}
\def\ADNI{\mbox{\tt A-DNI}}

\newcommand{\gbsecr}[4]{\mbox{\small $(#1)$}#2\mbox{\small $(#3\rightsquigarrow\!\!\![\!] #4)$}}
\newcommand{\gasecr}[4]{\mbox{\small $(#1)$}#2\mbox{\small $(#3\Rightarrow #4)$}}

\newcommand{\prabs}[2]{{#1}^{#2}}
\newcommand{\puabs}[2]{{#1}_{#2}}
\newcommand{\pruabs}[3]{{#1}_{#2}^{#3}}

\usepackage{amsbsy}
\usepackage{latexsym}
\usepackage{amsfonts}
\usepackage{amssymb}
\usepackage{amsmath}
\usepackage{stmaryrd}

\usepackage{url}

\usepackage{ifthen}

\usepackage{color}
\usepackage{graphics}
\usepackage{epsfig}
\usepackage{graphicx}
\usepackage{pstricks}

\usepackage{listings}

\lstdefinelanguage{pseudocode}
{morekeywords={data,do,else,foreach,function,procedure,if,return,returns,takes,then,until,while,repeat},
sensitive=false,
morecomment=[l]{//},
morestring=[b]'',}
\lstnewenvironment{pseudocode}{
\lstset{language=pseudocode}
\lstset{commentstyle=\textit}
\lstset{frame=tb}
\lstset{linewidth=83mm}
\lstset{xleftmargin=-3mm}
\lstset{frameround=ffff}
\lstset{framextopmargin=1mm}
\lstset{aboveskip=3mm}
\lstset{belowskip=4mm}
\lstset{framexbottommargin=1mm}
\lstset{framexleftmargin=-5mm}
\lstset{stringstyle=\underbar}
\lstset{mathescape=true}}{}

\newcommand{\COMMENT}[1]{}
\def\defemb#1#2{\expandafter\def\csname #1\endcsname{\relax\ifmmode #2 
\else\hbox{$#2$}\fi}}
\def\ok#1{\mbox{\raisebox{0ex}[1ex][1ex]{$#1$}}} 

\newcommand{\UNARYFUNCTION}[2]{#1\ifthenelse{\equal{#2}{}}{}{\left(#2\right)}}
\newcommand{\BINARYFUNCTION}[3]{#1\ifthenelse{\not\equal{#2}{}}{\left(#2\ifthenelse{\equal{#3}{}}{}{,#3}\right)}{}}
\newcommand{\TERNARYFUNCTION}[4]{#1\ifthenelse{\not\equal{#2}{}}{\left(#2\ifthenelse{\equal{#3}{}}{}{,#3}\ifthenelse{\equal{#4}{}}{}{,#4}\right)}{}}
\newcommand{\FOURARYFUNCTION}[5]{#1
\ifthenelse{\equal{#2}{}}{}{\left(#2\ifthenelse{\equal{#3}{}}{}{,#3}\ifthenelse{\equal{#4}{}}{}{,#4}\ifthenelse{\equal{#5}{}}{}{,#5}\right)}}
\newcommand{\UNARYFUNCTIONWITHSUBSCRIPT}[3]{#1\ifthenelse{\equal{#2}{}}{}{_{#2}}\ifthenelse{\equal{#3}{}}{}{\left(#3\right)}}
\newcommand{\BINARYFUNCTIONWITHSUBSCRIPT}[4]{#1\ifthenelse{\equal{#2}{}}{}{_{#2}}\ifthenelse{\equal{#3}{}}{}{\left(#3,#4\right)}}
\newcommand{\BINARYINFIXFUNCTION}[3]{\ifthenelse{\equal{#2}{}}{}{#2} #1 \ifthenelse{\equal{#3}{}}{}{#3}}

\newcommand{\CODE}[1]{\lstinline!#1!}

\newcommand{\PROOFSTEP}[1]{\ifthenelse{\equal{#1}{}}{}{[\mbox{#1}]}}

\def\cM{\mathcal{M}}

\defemb{tL}{{\tt L\:}}
\defemb{tH}{{\tt H\:}}

\title{Abstract interpretation-based approaches to Security\\{\large A Survey on Abstract Non-Interference and its Challenging Applications}}
\author{Isabella Mastroeni
\institute{Department of Computer Science}
\institute{University of Verona\\ Verona, Italy}
\email{isabella.mastroeni@univr.it}}

\def\titlerunning{Abstract non-interference}
\def\authorrunning{I. Mastroeni}
\begin{document}
\maketitle

\begin{abstract}
In this paper we provide a survey on the framework of abstract non-interference. In particular, we describe a general formalization of abstract non-interference by means of three dimensions (observation, protection and semantics) that can be instantiated in order to obtain well known or even new weakened non-interference properties. Then, we show that the notions of abstract non-interference introduced in language-based security are instances of this more general framework which allows to better understand the different components of a non-interference policy. Finally, we consider two challenging research fields concerning security where abstract non-interference seems a promising approach providing new perspectives and new solutions to open problems: Code injection and code obfuscation.
\end{abstract}

\section{Introduction}\label{intro}
Understanding information-flow is essential in code debugging, program
analysis, program transformation, and software verification but also in code protection and malware detection. Capturing
information-flow means modeling the properties of control and data
that are transformed dynamically at run-time. Program slicing needs
information-flow analysis for separating independent code; code
debugging and testing need models of information-flow for
understanding error propagation, language-based security needs
information-flow analysis for protecting data confidentiality from
erroneous or malicious attacks while data are processed by programs.
In code protection information-flow methods can be used for deciding where to focus the obfuscating techniques, while in malware detection information flow analyses can be used for understanding how the malware interact with its context or how syntactic metamorphic transformations interfere with the analysis capability of the malware detector.
The key aspect in information-flow analysis is understanding the
degree of independence of program objects, such as variables and
statements. This is precisely captured by the notion of {\em
  non-interference\/} introduced by Goguen and Meseguer
\cite{goguenmes82} in the context of the research on security
polices and models. 

The standard approach to language-based {\em
  non-interference\/} is based on a characterization of the attacker
that does not impose any observational or complexity restriction on
the attackers' power. This means that, the attackers
have {\em full power\/}, namely they are modeled without any
limitation in their quest to obtain confidential information. For this
reason non-interference, as defined in the literature, is an extremely
restrictive policy.  The problem of refining this kind of security
policy has been addressed by many authors as a major challenge in
language-based information-flow security \cite{SM03}.  Refining
security policies means weakening standard non-interference checks, in
such a way that these restrictions can be used in practice or can
reveal more information about how information flows in programs.

In the literature, we can find mainly two different approaches for
weakening non-interference: by constraining the power of the attacker
(from the observational or the computational point of view), or by
allowing some confidential information to flow (the so called {\em
  declassification\/}).  There are several works dealing with both
these approaches, but to the best of our knowledge, the first approach aiming at characterizing
at the same time both the power of the attacker's model
and the private information that can flow is {\em abstract non-interference} \cite{GM04popl} where the attacker is modeled as an abstraction of public data (input and output) and the information that may flow is an abstraction of confidential inputs. In this framework these two aspects are related by an adjunction relation \cite{GMadj10} formally proving that the more concrete the analysis the attacker can perform the less information we can keep protected.

In this paper, we introduce the abstract non-interference framework from a more general point of view. Data are simply partitioned in unobservable (called {\em internal}) and {\em observable} \cite{GMproof10} and  we may observe also relations between internal and observable data, and not simply attribute independent properties \cite{MB10}.
Our aim is that of showing that the abstract non-interference framework may be exported, due to its generality, to different fields of computer science, providing new perspectives for attacking both well known and new security challenges.

\paragraph*{Paper outline.}
The paper is structured as follows. In the following of this section we introduce the basic notions of abstract interpretation, abstract domain completeness and program semantics, used in the rest of the paper. In Sect.~\ref{secANI} we recall and we slightly generalize the notion of abstract non-interference ($\ANI$) formalized in the last years. In particular we describe $\ANI$ by means of three general dimensions: semantic, observation and protection. Finally we combine all these dimensions together. In Sect.~\ref{ANIsec} we provide a survey about how, in the literature, this notion of $\ANI$ has been used for characterizing weakened policies of non-interference in language-based security. Again, we organize the framework by means of the three dimensions that here become: {\em Who} attacks, {\em What} is disclosed and {\em Where/When} the attacker observes. Finally, we conclude the paper in Sect.~\ref{othSec} where we introduce two promising security fields where we believe the $\ANI$-based approach may be fruitful for providing a new perspective and a set of new formal tools for reasoning on challenging security-related open problems.

\paragraph*{Abstract interpretation: Domains and surroundings.}
%
Abstract interpretation is a general theory for specifying and designing approximate semantics of program languages \cite{CC77}. Approximation can be equivalently formulated
either in terms of Galois connections or closure operators
\cite{CC79}.  An {\em upper closure operator\/} $\rho: C\rightarrow C$ on a
poset $C$ ($\uco(C)$ for short), representing concrete objects, is monotone, idempotent, and extensive: $\forall x\in C.\;x\leq_C \rho (x)$.
The upper closure operator is the function that maps the concrete values to their abstract properties, namely with the best possible approximation of the concrete value in the abstract domain. For example, $\Sign:\wp(\nint)\ra\wp(\nint)$, on the powerset of integers, associates each set of integers with its sign: $\Sign(\varnothing)=\varnothing\defi${\em ``none''}, $\Sign(S)=\{n~|~n>0\}\defi+$ if $\forall n\in S.\:n>0$, $\Sign(0)=\{0\}\defi 0$, $\Sign(S)=\{n~|~n<0\}\defi-$ if $\forall n\in S.\:n<0$, $\Sign(S)=\{n~|~n\geq 0\}\defi 0+$ if $\forall n\in S.\:x\geq 0$, $\Sign(S)=\{n~|~n\leq 0\}\defi0-$ if $\forall n\in S.\:n\leq 0$ and $\Sign(S)=\nint\defi\mbox{\em ``I don't know''}$ otherwise. Analogously, the operator $\Par:\wp(\nint)\ra\wp(\nint)$ associates each set of integers with its parity, 
$\Par(\varnothing)=\varnothing\defi\mbox{\em ``none''}$, $\Par(S)=\{n\in\nint~|~n\ \mbox{is even}\}\defi\even$ if $\forall n\in S.\: n$ is even, $\Par(S)=\{n\in\nint~|~n\ \mbox{is odd}\}\defi\odd$ if $\forall n\in S.\:n$ is odd and $\Par(S)=\nint\defi\mbox{\em ``I don't know''}$ otherwise. 
Usually, {\em ``none''} and {\em ``I don't know''} are simply denoted $\varnothing$ and $\nint$.
%
Formally, closure operators $\rho$ are uniquely
determined by the set of their fix-points (or idempotents) $\rho (C)$, for instance $\Par=\{\nint,\even,\odd,\varnothing\}$. For upper
closures, $X\subseteq C$ is the set of fix-points of $\rho\in\uco(C)$
iff $X$ is a {\em Moore-family\/} of $C$, i.e., $X=\cM (X)\defi
\{\wedge S~|~ S\subseteq X\}$~---~where $\wedge \varnothing=\top \in
\cM (X)$. The set of all upper closure operators on $C$, denoted
$\uco(C)$, is isomorphic to
the so called {\em lattice of abstract interpretations of $C$\/}
\cite{CC79}.  If $C$ is a complete lattice then $\uco(C)$ ordered point-wise is also a complete lattice,  $\tuple{\uco(C),\sqsubseteq,\sqcup,\sqcap,\top,\id}$ where
for every $\rho,\eta \in \uco(C)$, $I\subseteq \Nats$, $\{ \rho_i \}_{i\in I} \subseteq
\uco(C)$ and $x\in C$: $\rho \sqsubseteq \eta$ iff $\forall y\in C.\;
\rho (y) \leq \eta(y)$ iff $\eta (C) \subseteq \rho (C)$;
$(\sqcap_{i\in I} \rho_i)(x) =\wedge_{i\in I} \rho_i (x)$; and
$(\sqcup_{i\in I} \rho_i)(x) =x \:\Lra\: \forall i\in I.\; \rho_i (x)
= x$.  

Abstract interpretation is a theory for approximating program behaviour by approximating their semantics. Now, we formally introduce the notion of precision in terms of abstract domain completeness. There are two kinds of completeness, called {\em backward} and {\em forward} completeness \cite{GQ01}.  Backward completeness ($\cB$) requires accuracy when we compare the computations on the program input domain: the abstract outputs of the concrete computation $f$ are the same abstract outputs obtained by computing the program on the abstract values.
Formally, $\rho$ is backward complete for $f$ iff $\rho\comp f\comp\rho=\rho\comp f$ \cite{CC79}. 
Consider Fig.~\ref{bcomplet}. The outer oval always represents the concrete domain, while the inner one represents the abstract domain characterised by the closure $\rho$. The computation is represented by the function $f$. Hence, on the left, we have incompleteness since the 
abstract computation on the abstract values ($\rho(f(\rho(x)))$) loses precision with respect to  (i.e., is more abstract than) the abstraction of the computation on the concrete values ($\rho(f(x))$). On the right, we have  completeness because the two abstract computations coincide.
\begin{figure}
\begin{center}
\includegraphics{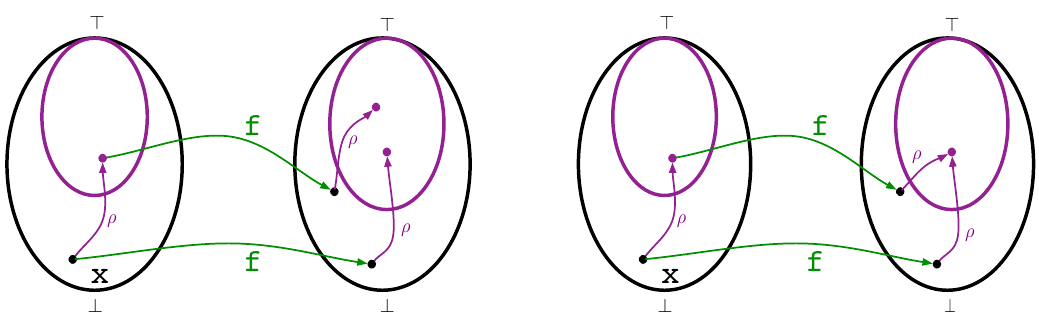}
\end{center}
\caption{Backward completeness}\label{bcomplet}
\end{figure}
Forward completeness ($\cF$) requires accuracy  when we compare the abstract and the concrete computations on the  output domain of the program, i.e., we compare whether the abstract and the concrete outputs are the same when the program computes on the abstract values.  Formally, given a semantics $f$ and a closure $\rho$, $\rho$ is forward complete for $f$ iff $\rho\comp f\comp \rho=f\comp\rho$.
Consider Fig.~\ref{fcomplet} . On the left, we have incompleteness since the concrete and the abstract computations on abstract values (respectively $f(\rho(x))$ and $\rho(f(\rho(x)))$) does not provide the same result, and in particular the abstraction of the computation loses precision. On the right, the two computations coincide since $f$ returns, as output, an element in $\rho$, and therefore we have completeness.
\begin{figure}
\begin{center}
\includegraphics{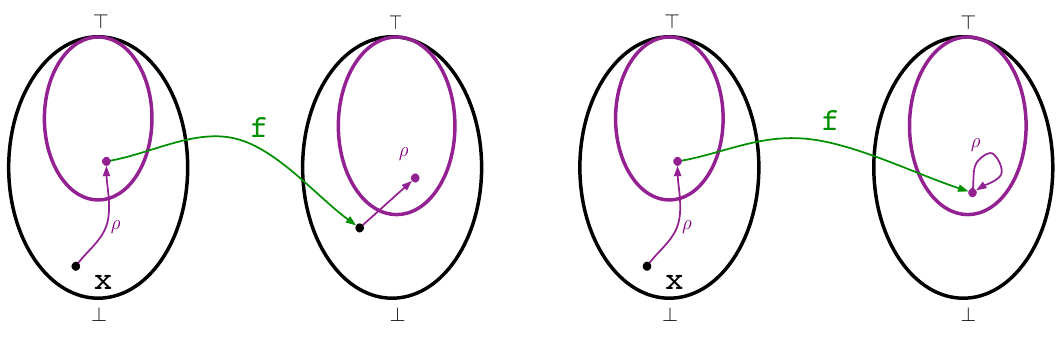}
\end{center}
\caption{Forward completeness}\label{fcomplet}
\end{figure}
Finally we observe that, if $f$ is additive\footnote{It commutes with least upper bounds.}, then there exists $f^+\defi\lambda x.\bigvee\{~y|~f(y)\leq x\}$ and we have that $\rho$ is $\cB$-complete for $f$ iff it is $\cF$-complete for $f^+$.

\paragraph*{Programs and Semantics.} Consider 
a simple (deterministic) imperative language $\cL$: $
\tC  ::=  \nil ~|~ x:=e~|~ \tC_0;\tC_1~|~ \while\ B\ \dow\ \tC\ \ew~|~
\ifc\ B\ \thenc\ \tC_0\ \elsec\ \tC_1
$.
Let $\prog{\cL}$ be a set of programs in the language $\cL$,
$\Var(P)$ the set of all the variables in $P\in\prog{\cL}$, 
$\val$ be the set of values, $\tM:\Var(P)\ra \val$. 
In sake of simplicity, in the examples, we consider a fine-grained big-step operational semantics, where each step of computation corresponds to the execution of a single statement, e.g. the execution of an $\ifc$ statement corresponds to one step. In this way, the program points coincide with the steps of computation. This is not a mandatory choice and the whole framework can be extended to trace semantics of any granularity.
Let $\tpost$ be the transition relation induced by the fine-grained big-step semantics, then we denote by $\grasstr{P}$ the set of execution traces of $P$ obtained by fix-point iteration of the transition relation \cite{C00tcs}. We denote by $\grass{P}:\tM\ra\tM$ the denotational semantics of $P$ associating with initial states the corresponding final states. In other words $\grass{P}$ is the I/O abstraction of the trace semantics $\grasstr{P}$ \cite{C00tcs}.

\section{A general framework for Abstract Non-Interference}\label{secANI}
In this section, we introduce the notion of abstract non-interference
\cite{GM04popl}, i.e., a weakening of non-interference given
by means of abstract interpretations of
concrete semantics.  We will start from standard notion on
non-interference ($\NI$ for short), originally introduced in
language-based security \cite{cohen77,goguenmes82,SM03}, and here
generalized to any kind of classification of data (intended as program variables), where we are
interested in understanding whether a given class of data ({\em internal}) interferes with
another class of data ({\em observable}). In other words, we generalize the
public/private data classification in language-based security
to a generic observable/internal classification \cite{GMproof10}.
\\
Consider variables {\em statically}\footnote{In this paper we do not consider security types that can dynamically change.} distinguished into  {\em
  internal} (denoted $\rel$) and {\em observable} (denoted $\obs$).  The
internal data\footnote{In sake of simplicity, we identify data with their containers (variables).} correspond to those variables that must not interfere
with the observable ones.  This partition characterises the $\NI$ policy we have to
verify/define.
\\
Both the input and the output variables are partitioned in this way,
and the two partitions need not coincide. Hence, if $\Pinp$ denotes the set of input variables and $\Poup$ denotes the set of output variables, we
have four classes of data: $\Prelin$ are the input {\em internal} variables, $\Pobsin$ are the not internal (potentially observable) input, $\Prelout$
are the not observable (hence internal) outputs, and
$\Pobsout$ are the output {\em observables}. Note that the formal distinction between $\Pinp$ and $\Poup$ is used only to underline that there can be different partitions of input and output, namely $\Prelin=\Prelout$ need not hold. In general, we have the same set of variables in input and in output, hence $\Pinp=\Poup=\Var(P)$.\\
Informally, non-interference
can be reformulated by saying that if we fix the values of variables in $\Pobsin$ and we let
values of variables in $\Prelin$ change, we must not observe any difference in the values of variables in $\Pobsout$. Indeed if this happens it means that $\Prelin$ interferes
with $\Pobsout$.  We will use the following notation: if $n=|\{x\in\Var(P)~|~x\ \mbox{is internal}\}|=|\Prelin|$, then $\relin\defi\val^{n}$, analogously 
$\relout\defi\val^{|\Prelout|}$, $\obsin\defi\val^{|\Pobsin|}$ and $\obsout\defi\val^{|\Pobsout|}$, 
where 
$|X|$ denotes the cardinality of the set of variables $X$. Consider $\mtC\in\{\relin,\obsin,\relout,\obsout\}$,
in the following, we abuse notation by denoting $v\in\mtC$ the
fact that $v$ is a possible tuple of values for the vector of variables evaluated in $\mtC$, e.g., $v\in\relin$ is a vector of values for the variables in $\Prelin$. Moreover, if $x$ is a tuple of variables in $\Poup$ (analogous for $\Pinp$) we denote as $x^{\rel}$ [resp.\ $x^{\obs}$] the projection of the tuple of variables $x$ only on
the variables in $\Prelout$ [resp.\ $\Pobsout$] (analogous for values). At this point, we can reformulate standard
non-interference for a deterministic program $P$\footnote{If $P$ is not deterministic the definition works anyway simply by interpreting $\grass{P}(s)$ as the set of all the possible outputs starting from $s$.}, w.r.t.\ fixed partitions of input and output variables
$\pi_{\tI}\defi\{\Pobsin,\Prelin\}$ and $\pi_{\tO}\defi\{\Pobsout,\Prelout\}$ ($\pi=\{\pi_\tI,\pi_\tO\}$):
\vspace{-.5cm}
\begin{center}\label{NIdef}
\begin{eqnarray}\label{defNI}
\mbox{
\framebox{
\begin{tabular}{c}
A program $P$, satisfies {\em non-interference} w.r.t. $\pi$ if \\
$\forall v\in\obsin,\forall v_1,v_2\in\relin\:.\:
(\grass{P}(v_1,v))^{\obs}=(\grass{P}(v_2,v))^{\obs}$
\end{tabular}}}
\end{eqnarray}
\end{center}

%
\subsection{An abstract domain completeness problem}
In this section, we recall from \cite{GMadj10,MB10} the completeness formalization of (abstract) non-interference. This characterization underlines that non-interference holds when the abstraction of input and output data (the projection on observable values in standard $\NI$) is complete, i.e., precise, w.r.t.\ the semantics of the program. This exactly means that, starting from a fixed input property the semantics of the program does not change the output observable property \cite{HM05}.\\
Joshi and Leino's characterization of classic $\NI$ in  \cite{joshiLeino00} provides an equational definition of $\NI$ which can be easily rewritten as a completeness equation: a program $P$ containing internal
and observable variables (ranged over by $p$ and $o$ respectively) satisfies non-interference
iff $HH;P;HH = P;HH$,
where $HH$ is an assignment of an arbitrary
value to $p$. ``The postfix occurrences of $HH$ on each side mean that
we are only interested in the final value of $o$ and the prefix $HH$
on the left-hand-side means that the two programs are equal if the
final value of $o$ does not depend on the initial value of
$p$''~\cite{SS-HOSC01}.

An abstract interpretation is (backwards) complete for a function, $f$,
if the result obtained when $f$ is applied to any concrete input, $x$,
and the result obtained when $f$ is applied to an abstraction of the
concrete input, $x$, both abstract to the same value. The completeness connection is implicit in Joshi and
Leino's definition of secure information flow and the implicit
abstraction in their definition is: ``each internal value is associated
with $\top$, that is, the set of all possible internal values''. 

Let  $\pi\defi\pi_\tI=\pi_\tO$, namely $\relin=\relout, \obsin=\obsout$. The set of program states is $\Sigma =
\relin\times\obsin$, which is implicitly indexed by the internal
variables followed by the observable variables. 
%

Because $HH$ is an arbitrary assignment to $p\in\relin$, its 
semantics can be modelled as an \emph{abstraction function}, $\cH$, on sets of concrete
program states, $\Sigma$; that is, $\cH:\wp(\Sigma)\to\wp(\Sigma)$,
where $\wp(\Sigma)$ is ordered by subset inclusion, $\subseteq$.
For each possible value of an observable variable, $\cH$ associates
\emph{all} possible values of the internal variables in $P$.  Thus $\cH(X)
= \relin \times X^\obs$, where $\relin$ is the top element of
$\wp(\relin)$.  Hence the Joshi-Leino definition can be
rewritten~\cite{GMadj10} in the following way:
\begin{eqnarray}
\label{eqn:bc}
\cH\circ\grass{P}\circ\cH = \cH\circ \grass{P}
\end{eqnarray}
It is clear that $\cH$ is parametric on the partition $\pi$ (and in general the abstraction $\cH$ applied to the input is parametric on $\pi_\tI$, while the one applied on the output is parametric on $\pi_\tO$). Anyway, in sake of readability we use simply $\cH$ instead of $\cH_\pi$.\\
The equation above is precisely the definition of backwards completeness in abstract interpretation~\cite{CC79,GRS00}
(see \cite{BM08} for examples).
Note that, Equation~(\ref{eqn:bc}) gives us a way to
\emph{dynamically} check whether a program satisfies a confidentiality policy, this is due to the use of denotational semantics. In \cite{BM08} we show that we can perform the same analysis statically, by involving the weakest precondition semantics.

In particular, static checking involves
$\cF$-completeness, instead of $\cB$-completeness, and the use of
weakest preconditions instead of the denotational semantics \cite{BM08}.
With weakest preconditions, (written $\Wlp_P$), equation
~(\ref{eqn:bc}) has the following equivalent reformulation:
\begin{eqnarray}
\label{eqn:fc}
\cH\circ \Wlp_P\circ \cH = \Wlp_P\circ\cH
\end{eqnarray}
Equation~(\ref{eqn:fc}) says that $\cH$ is $\cF$-complete for
$\Wlp_P$. The equation asserts
that $\Wlp_P(\cH(X))$ is a fixpoint of $\cH$, meaning that  $\Wlp_P\comp\cH$
associates each observable output with any possible internal input: a further abstraction 
of the fixpoint (cf., the lhs of equation~(\ref{eqn:fc})) yields nothing 
new. Because no \emph{distinctions among internal inputs} get exposed to 
an observer, the observable output is independent of the internal input, hence also equation~(\ref{eqn:fc}) asserts classic $\NI$.

%
\subsection{Tuning Non-Interference: Three dimensions of Non-Interference} \label{ani}
We believe that the real added value of abstract non-interference is the possibility of deeply understanding which are the {\em actors} playing and which is their role in the definition of a security policy to enforce \cite{BM08}. In particular, we can observe that in the abstract non-interference framework we can identify three {\em dimensions}: (a) the semantic dimension, (b) the
observation dimension and (c) the protection/declassification
dimension. These dimensions are pictorially represented in Fig.~\ref{geo} \cite{MB10}.
\begin{figure}[htbp]
\begin{center}
\includegraphics[scale=.65]{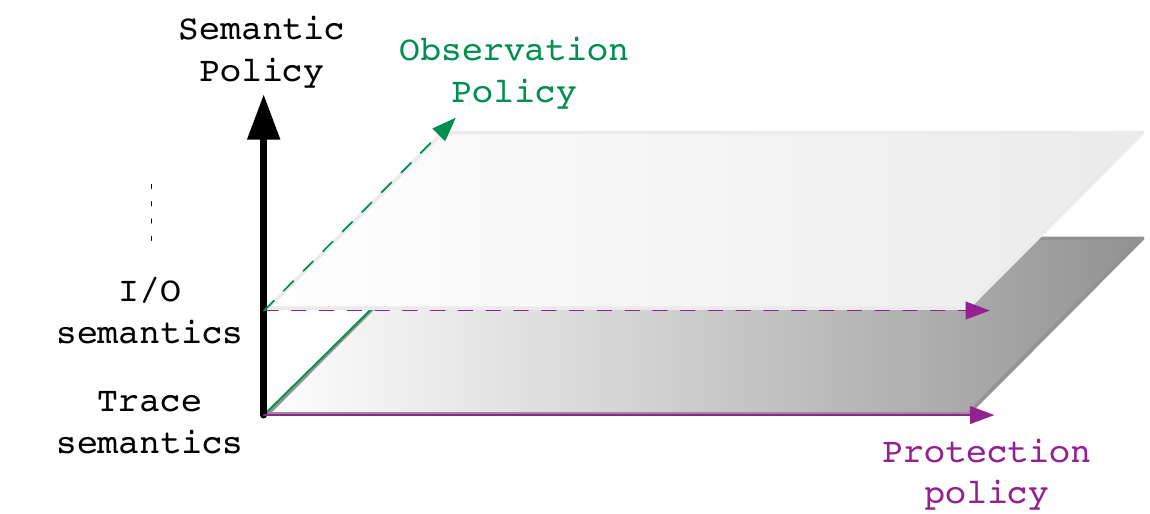}
\end{center}
\caption{The three dimensions of non-interference}\label{geo}
\end{figure}
In general, to describe any non-interference property, we first fix its semantic dimension. The semantic dimension comprises of the \emph{concrete semantics} of the program (the solid arrow in Fig.~\ref{geo} shows Cousot's hierarchy of semantics \cite{C00tcs}), the \emph{set of observation points}, that is the program points where the observer can analyse data, and the \emph{set of protection points}, that is the program points where information must be protected. Next we fix the program's observation and protection dimensions that say \emph{what} information can be observed and \emph{what} information needs to be protected at these program points.   This is how the \emph{what} dimension of declassification policies \cite{SSJCS07} is interpreted in \cite{BM08}.
\\
Consider, for instance standard $\NI$ in Eq.~\ref{defNI}.
The three dimensions of $\NI$ are as follows. The semantic dimension, i.e., the concrete  semantics, is $P$'s denotational semantics. Both inputs and outputs constitute the observation points while inputs constitute the protection points because it is internal data at program inputs that need to be protected.
 The observation dimension is
the identity on the observable input-output because the observer can only analyse
observable inputs and outputs. Such an observer is the most
powerful one for the chosen concrete semantics because its
knowledge (that is, the observable projections of the initial and final
states of the two runs of the program) is completely given by the
concrete semantics.  Finally, the protection dimension is the
identity on the \emph{internal input}. This means that \emph{all} internal
inputs must be protected, or dually, \emph{no} internal inputs must be
released.\\
%
In this survey we generalize these dimensions in the notion of abstract non-interference introduced in \cite{GMproof10}, which goes beyond the standard low and high data classification in language-based security.

\subsubsection{The semantic dimension}\label{sec:sempol}

In this section our goal is to provide a general description of a semantic dimension that is parametric on any set of protection and observation points. For this purpose it is natural to move to a trace semantics. The first step in this direction is to consider a function $\post$ as defined below.
Let us define
$
\post_{P}\defi
\{\tuple{s,t}~|~s,t\in\Sigma,\ s\tpost t\},
$
where, we recall that $\tpost$ is the semantic transition relation.
From this definition of $\post_P$ it is quite straightforward to recover non-interference based on I/O observation \cite{MB10}. Let $\post^{+}_{P}$ the transitive closure of $\post_{P}$ associating with each initial state $s$ the final state $t$ reachable from $s$, i.e., $\post^+_P\defi\{\tuple{s,t}~|~s,t\in\Sigma,\ s\tpost^* t,\ t\ \mbox{final}\}$.
Then, an equivalent characterization of $\NI$, where $\Sigma_{\vdash}$ is the set of initial states, is
\[
\forall s_{1},s_{2}\in\Sigma_{\vdash}.\:s_{1}^\obs=s_{2}^\obs\ \Ra\ \post^{+}_{P}(s_{1})^\obs=\post^{+}_{P}(s_{2})^\obs
\]

\paragraph{\em $\NI$ for trace semantics.}
A denotational semantics does not take into account the whole history
of computation, and thus restricts the kind of
protection/declassification policies one can model.  In order to
handle more precise policies, that take into account {\em where/when}
\cite{SSJCS07} information is released in addition to {\em what}
information is released we must consider a more concrete semantics,
such as trace semantics. More precisely, depending on how we fix the observation points we describe a {\em where} or a {\em when} dimension: If the observation points are program points then we are fixing {\em where} to observe, if the they are computational steps then we are fixing {\em when} to observe. In our semantics these points coincide and we choose to call this dimension {\em where}.\\
Let us define $\NI$ on traces:
\[
\forall  s_{1},s_{2}\in\Sigma_{\vdash}.\:s_{1}^\obs=s_{2}^\obs\ \Ra\ \grasstr{P}(s_{1})^{\obs}=\grasstr{P}(s_{2})^{\obs}
\]
This definition says that given two observably indistinguishable input
states, $s_1$ and $s_2$, the two executions of $P$ must
generate two --- both finite or both infinite --- sequences of states
in which the corresponding states in each sequence are observably
indistinguishable. 
Equivalently, we can use a set of $\post$ relations: for $i\in \Nats$ we define the family of relations
$\post^{i}_{P}\defi\{\tuple{s,t}~|~t\in\Sigma,\ s\in\Sigma_{\vdash},\ s\tpost^{i}t\}$, 
i.e., $\post^{i}_{P}$ is the I/O semantics after $i$ steps of computations.
The following result is straightforward.
\begin{proposition}\cite{MB10}
$\NI$ on traces w.r.t. $\pi_\tI$ and $\pi_\tO$ holds iff for each program point, $i$, of program $P$, we have 
$
\forall s_{1},s_{2}\in\Sigma_{\vdash}.s_{1}^\obs=s_{2}^\obs\ \Ra\ \post^{i}_{P}(s_{1})^{\obs}=\post^{i}_{P}(s_{2})^{\obs}
$.
\end{proposition}

The $\post$ characterization of trace-based $\NI$
precisely identifies the observation points as the outputs of the
$\post$ relations, that is, any possible intermediate state of
computation, and the protection points are identified as the inputs of the $\post$
relations, that is, the initial states.

\paragraph{\em General semantic policies.}
In the previous paragraph, we show how for denotational and trace
semantics, we can define a corresponding set of $\post$ relations
fixing protection and observation points. In order to understand how
we can generalize this definition, let us consider a graphical
representation of the situations considered in Fig.~\ref{fig:scheme1} \cite{BM08}.
\begin{figure}[htbp]
\begin{center}
\includegraphics[scale=.6]{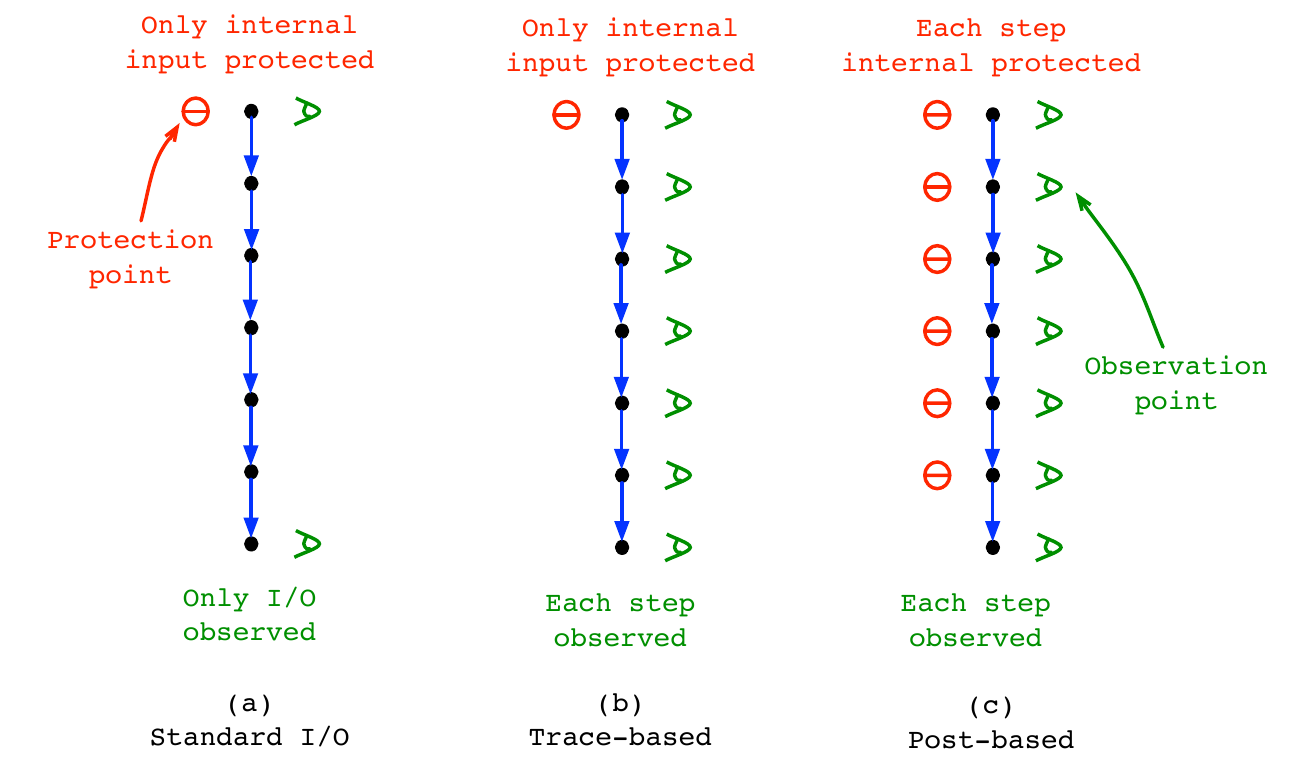}
\end{center}
\caption{Different notions of non-interference for trace semantics}
\label{fig:scheme1}
\end{figure}
\begin{description}
\item[(i)] In the first picture, the semantic dimension says that an
  observer can analyse the observable inputs and outputs only, while we
  can protect only the internal inputs. This notion corresponds to $\NI$.
\item[(ii)] In the second picture, the semantic dimension says that an
  observer can analyse each intermediate state of computation,
  including the input and the output, while the protection point is
  again the input.  
\item[(iii)] In the last picture, the semantic dimension says that an observer can analyse each
  intermediate state of computation, while the protection points are
  all intermediate states of the computation. In order to check this
  notion of non-interference for a program we have to check
  non-interference separately for each statement of the program
  itself. It is worth noting that this corresponds exactly to $\forall
  s_{1},s_{2}\in\Sigma.\:s_{1}^{\obs}=s_{2}^\obs.\:\post_{P}(s_{1})^\obs=\post_{P}(s_{2})^\obs$.
\end{description}
It is clear that between (i) and (ii) there are several notions of
non-interference depending on the intermediate observable states fixed by the where dimension of the policy
({\em observation points} in Fig.~\ref{fig:scheme1}).
For example, in language-based security, gradual release \cite{AS07} considers as observation points only those program points corresponding to observable events (i.e., assignment to observable variables, declassification points and termination points). 
However, unless we consider interactive systems that
can provide inputs in arbitrary states, we will
protect only the initial ones. Hence, in this case, (ii) and (iii) collapse to the same notion.
\begin{definition}[Trace-based $\NI$]\label{def:sempol}~\cite{MB10}
Given a set $\tO$ of observation points, the notion of $\NI$ w.r.t. $\pi_\tI$ and $\pi_\tO$, based on the
semantic dimension w.r.t.\ $\tO$, and the input as
protection point, is defined:
\[
\begin{array}{c}
\forall j\in\tO.\:\forall s_{1},s_{2}\in \Sigma_{\vdash}.\:
 s_{1}^\obs=s_{2}^\obs\ \Ra\ \post^{j}_{P}(s_{1})^\obs=\post^{j}_{P}(s_{2})^\obs
\end{array}
\]
\end{definition}
Note that, this is a general notion, that can be formulated depending on $\tO$. In particular, we obtain standard $\NI$ by fixing $\tO=\{p~|~p\ \mbox{is the final program point}\}$, we obtain the most concrete trace-based $\NI$ policy by fixing $\tO=\{p~|~p\ \mbox{is any program point}\}$. But we can also obtain intermediate notions depending on $\tO$: we obtain gradual release by fixing $\tO=\{p~|~p\ \mbox{is a program point corresponding to an observable event}\}$. 

\begin{example}\label{runningEx}
Let us consider the program fragment $P$ with the semantic policies represented in the following picture. The semantic dimension represented in picture $(a)$ is the I/O one. In this case, for each pair of initial states $s_{1},s_{2}$ such that $s_{1}^{\obs}=s_{2}^{\obs}$ we have to check $\post^{6}(s_{1})^{\obs}=\post^{6}(s_{2})^{\obs}$, and it is clear that this hold since, for instance, both $l_1$ and $l_2$ become constant.
\[
\begin{array}{lll}
 {\dred p_{0}}\ra&&\la{\dgreen{o_{0}}}\\
& h_{1}:=h_{2};\\
&h_{2}:=h_{2}\modu 2;\\
&l_{1}:=h_{2}\\
&l_{2}:=h_{1}\\
&l_{2}:=5\\
&l_{1}:=l_2+3\\
&&\la{\dgreen{o_{6}}}\\ 
&(a)
\end{array}
\hspace{1.5cm}
\begin{array}{lll}
 {\dred p_{0}}\ra&&\la{\dgreen{o_{0}}}\\
& h_{1}:=h_{2};\\
&h_{2}:=h_{2}\modu 2;\\
&l_{1}:=h_{2}\\
&&\la{\dgreen{o_{3}}}\\
&l_{2}:=h_{1}\\
&&\la{\dgreen{o_{4}}}\\
&l_{2}:=5\\
&l_{1}:=l_2+3\\
&&\la{\dgreen{o_{6}}}\\ 
&(b)
\end{array}
\] 
On the other hand, the policy in picture $(b)$ considers $\tO=\{3,4,6\}$ (the observable events). In this case, for each pair of initial states $s_{1},s_{2}$ such that $s_{1}^{\obs}=s_{2}^{\obs}$ we have to check $\post^{6}(s_{1})^{\obs}=\post^{6}(s_{2})^{\obs}$, but also $\post^{4}(s_{1})^{\obs}=\post^{4}(s_{2})^{\obs}$
and $\post^{3}(s_{1})^{\obs}=\post^{3}(s_{2})^{\obs}$, and both these new tests fail since in all the corresponding program points there is a leakage of private information. For instance, if $s_{1}=\tuple{h_{1},h_{2}=2,l_{1},l_{2}}$ and $s_{2}=\tuple{h_{1},h_{2}=3,l_{1},l_{2}}$ then we have different $l_{1}$ values in $o_3$: respectively $l_{1}=0$ and $l_{1}=1$.
\end{example}

\paragraph{\em Semantic dimension in the completeness formalization.} We can observe that the completeness equation is parametric on the semantic function used. Since the denotational semantics is the $\post$ of
a transition system where all traces are two-states long. A generalization of the completeness reformulation, in order to cope also with trace-based $\NI$ (Def.~\ref{def:sempol}),
can be immediately obtained \cite{MB10}.  Theorem~\ref{mcforAni} below 
shows the connection, also for traces, between completeness and non-interference.
\begin{theorem}\label{mcforAni}\cite{MB10}
  Let $\tuple{\Sigma,\tpost}$ be the transition system for a program $P$, and $\tO$ a set of observation points. Then, trace-based $\NI$ w.r.t.\  $\pi_\tI$ and $\pi_\tO$ (Def.~\ref{def:sempol})
holds iff
  $\forall j\in\tO.\:\cH\comp\post^{j}_P\comp\cH=\cH\comp\post^{j}_P$.
\end{theorem}
 
This theorem characterizes $\NI$ as a family of completeness problems also when a malicious attacker can potentially observe the whole trace semantics, namely when we deal with trace-based $\NI$.

Moreover, let us denote as
$\widetilde{\pre}^{j}$ the adjoint map of $\post^{j}$\footnote{By adjoint relation the function $\widetilde{\pre}^{j}$ is the weakest precondition of the corresponding function $\post^{j}$} in the same
transition system, then the completeness equation can be rewritten as
$\cH\comp\widetilde{\pre}^{j}\comp\cH=\widetilde{\pre}^{j}\comp\cH$. 
%

\subsubsection{The observation dimension} \label{sec:obspol}

Consider the program $P\defi x:= |x|*\Sign(y)$, where $\Prelin=\{y\}$ and $\Pobsin=\Pobsout=\{x\}$,
suppose that $|\cdot|$ is the absolute value function, then ``{\em only a portion of $x$
  is affected, in this case $x$'s sign. Imagine if an observer could
  only observe $x$'s absolute value and not $x$'s sign\/}''
\cite{cohen77} then we could say that in the program there is
non-interference between $\rel$ and $\obs$. Abstract interpretation
provides the most appropriate framework to further develop Cohen's
intuition. The basic idea is that an observer can analyze only some
properties, modeled as abstract interpretations of the concrete
program semantics.


Suppose the observation points fixed by the chosen semantics are input
and output.  Then the observation dimension might require that the
observer analyse a particular \emph{property}, $\rho$, of the observable
output --- e.g., parity --- and a particular property, $\eta$, of the
observable input --- e.g., signs.  
In the following, we will consider
$\eta\in\uco(\inp)$ such that $\eta$ abstracts
internal and observable variables. This abstraction may be attribute independent
or relational\footnote{Here, by
  relational, we mean not attribute independent, namely a property
  describing relations of elements, for example $\eta(\tuple{x,y})=0+$
  if $x+y\geq 0$ is a relational property.}. Attribute independent means that $\eta$ can be split in two independent abstractions identifying precisely what is observable ($\obsin$) or not ($\relin$). Hence, in this case it can be split in one abstraction
for the variables in $\Prelin$, denoted $\eta_{\rel}$, and one for the
variables in $\Pobsin$, denoted $\eta_{\obs}$, and we write
$\eta=\tuple{\eta_{\rel},\eta_{\obs}}$. 
Then we obtain a weakening of standard $\NI$ as
follows:
\begin{equation}
\label{eqn:nni}
\forall x_1,x_2\in\val\:.\:
x_1^{\obs}=x_2^{\obs}\ \Ra\ 
{\dred \rho}(\grass{P}({\dred \eta}(x_1)))= {\dred \rho}(\grass{P}({\dred \eta}(x_2)))
\end{equation}
This weakening, here called
\emph{abstract non-interference ($\ANI$)}, was first partially introduced in~\cite{GM04popl}. The interesting cases are two. The first is $\eta=\id$, which consists in a {\em dynamic} analysis where the observers collects the (possibly huge) set of possible computations and extract in some way (e.g. data mining) properties of interest. The second is an observer performing a {\em static} analysis of the code, in which case usually $\eta=\rho$.
Standard $\NI$ is recovered by setting
$\eta$ to be the identity and $\rho$ to the projection on observable values, i.e., $\rho=\tuple{\topi,\ido}$, where $\topi\defi\lambda x\in\relin.\:\top$ (in this case $\top=\relin$) and $\ido=\lambda x\in\obsin.\:x$ is the identity on observables.

\paragraph{\em Observation dimension for a generic semantic dimension.} 
In \cite{BM08} we showed how to combine a semantic dimension that comprises of 
a trace-based semantics with an
observation dimension. In other words, we showed how we abstract a trace by a state abstraction.
Consider the following concrete trace where each state in the trace is represented as a
pair $\tuple{x^\rel,x^\obs}$.
\[
\tuple{3,1} \to \tuple{2,2} \to \tuple{1,3} \to \tuple{0,4} \to \tuple{0,4}
\]
Suppose also that the trace semantics fixes the observation points to
be each intermediate state of computation.  Now suppose the
observation dimension is that only the parity (represented by the
abstract domain, $\Par$) of the public data can be observed. Then the
observation of the above trace through $\Par$ is:
\[
\tuple{3,odd} \to \tuple{2,even} \to \tuple{1,odd} \to \tuple{0,even} \to \tuple{0,even}
\]
The abstract notion of $\NI$ on traces is formulated by saying that all the
execution traces of a program starting from states with the same property ($\eta$) of public input,
have to provide the same property ($\rho$) of
reachable states.
Therefore, the general notion of $\ANI$ consists simply in abstracting each state
of the computational trace. We thus have
\begin{definition}[Trace-based $\ANI$]\label{def:obspol}
Given a set of observation points $\tO$ and the partitions $\pi_\tI$ and $\pi_\tO$
\[
\forall j\in\tO.\:\forall s_1,s_2\in\Sigma_{\vdash}\:.\:
s_1^{\obs}=s_2^{\obs}\ \Ra\ 
{\dred \rho}(\post^{j}_{P}({\dred \eta}(s_1)))= {\dred \rho}(\post^{j}_{P}({\dred \eta}(s_2)))
\]
\end{definition}
The following example shows the meaning of the observation dimension.
\begin{example}
Consider the program fragment in Ex.~\ref{runningEx} together with the semantic policies shown so far. If we consider the semantic dimension in $(a)$ and an observer able only to analyse in output the sign of integer variables, i.e., $\eta=\id$ and $\rho=\tuple{\topi,\rho_\obs}$, $\rho_\obs=\{\varnothing,<0,\geq 0,\top\}$, trivially we have that, for any pair of initial states $s_{1}$ and $s_{2}$ agreeing on the observable part, $\rho(\post^{6}(s_{1}))=(\geq 0)=\rho(\post^{6}(s_{2}))$. Namely, the program is secure. Consider now the semantic dimension in $(b)$ and the same observation dimension. In this case, non-interference is still satisfied in $o_{6}$ but it fails in $o_{3}$ and in $o_{4}$ since by changing the sign of $h_{2}$ we change the sign of respectively $l_{1}$ and $l_{2}$.  
\end{example}


\paragraph{\em Observation dimension in the completeness formalization.}
In order to model $\ANI$ by using the completeness equation we have to embed the abstractions characterizing the attacker into the abstract domain $\cH$. Hence, let us first note that $\cH=\lambda
X.\tuple{\topi(X^\rel),\id^\obs(X^\obs)}=\lambda X.\:\tuple{\relin,X^\obs}$
where $X^\rel\defi \{x^\rel~|~x\in X\}$, i.e., $\cH$ is the
product of respectively the top and the bottom abstractions in the
lattice of abstract interpretations. 
As far as the semantics is concerned, if the attacker may perform a static analysis w.r.t.\ $\eta$ in input and $\rho$ in output, then the semantics is its best correct approximation, i.e., 
$$\pruabs{\fP}{\eta}{\rho}\defi\lambda x.\:\rho\comp \fP\comp\eta(x)$$  
This means that the right formalization of $\ANI$ via completeness is
\begin{eqnarray}
\puabs{\cH}{}\comp
\pruabs{\fP}{\eta}{\rho}\comp\puabs{\cH}{}=\puabs{\cH}{}\comp \pruabs{\fP}{\eta}{\rho}
\end{eqnarray}
The next theorem shows that the equation above completely characterizes $\ANI$ as a completeness problem. This theorem is a generalization of the one proved for language-based security \cite{GMadj10}.
\begin{theorem}\label{thANI1}\label{thANI2}
Consider $\rho\in\uco(\wp(\inp))$ defining what is observable, $\eta\in\uco(\wp(\inp))$:
\[
P\ \mbox{satisfies $\ANI$ in Eq.~\ref{eqn:nni}}\ \Lra\ \puabs{\cH}{}\comp
\pruabs{\fP}{\eta}{\rho}\comp\cH=\puabs{\cH}{}\comp \pruabs{\fP}{\eta}{\rho}.
\]
\end{theorem}
If $\eta=\id$, and 
$\rho=\tuple{\rho_\rel,\rho_\obs}\in\uco(\wp(\inp))$, we define
$\puabs{\cH}{\rho}\in\uco(\wp(\inp))$:
$\puabs{\cH}{\rho}\defi\lambda X.\:\cH(X)\comp \rho=\tuple{\relin,\rho_\obs(X^\obs)}\in\uco(\wp(\inp))$. In this case we can rewrite the completeness characterization as
\begin{eqnarray}
\puabs{\cH}{\rho}\comp
\puabs{\fP}{}\comp\puabs{\cH}{}=\puabs{\cH}{\rho}\comp \puabs{\fP}{}
\end{eqnarray}

\subsubsection{The protection dimension}\label{sec:pdp}

Suppose now, we aim at describing a
property on the input representing for which
inputs, we are interested in testing $\NI$ properties \cite{MB10}.
Let $\Sel\in\uco(\wp(\inp))$ such a property, also this input property can be modeled, when possible, as an attribute independent abstraction of states, i.e., $\phi=\tuple{\phi_\rel,\phi_\obs}$.
This component of the $\NI$ dimension specifies {\em what} 
must be protected or dually, what must be declassified. \\
For denotational semantics, standard $\NI$ says that nothing must be
declassified.  Formally, 
we can say that the property $\topi$ has been declassified. 
From the observer perspective, this means that every internal input has
been mapped to $\top$. On the other hand, suppose we want to
declassify the parity of the internal inputs.  Then we do not
care if the observer can analyse any change due to the variation of parity
of internal inputs. For this reason, we only check the variations
of the output when the internal inputs have the same parity
property. Formally, consider an abstract domain $\phi$ --- the selector,
that is a function that maps any internal input to its corresponding parity \footnote{More precisely, $\phi$ maps a set of internal inputs to the join
of the parities obtained by applying $\phi$ to each internal input; the join
of $even$ and $odd$ is $\top$.}.
Then a program $P$ satisfies declassified non-interference ($\DNI$) provided
\begin{equation}
\label{eqn:dni}
\forall x_1,x_2\in\val\:.\:
{\dred \phi}(x_{1})={\dred \phi}(x_{2})\ \Ra\ 
\grass{P}(x_1)^\obs= \grass{P}(x_2)^\obs
\end{equation}
This notion of $\DNI$ has been advanced several times in the literature,
e.g., by~\cite{SM04}; this particular formulation using abstract
interpretation where $\phi\defi\tuple{\phi_\rel,\id^\obs}$ is due to~\cite{GM04popl} and is further explained in
~\cite{M05Aplas} where it is termed ``declassification by allowing''.
The generalization of $\DNI$ to Def.~\ref{def:sempol} is
straightforward. Since we protect the internal inputs, we
simply add the condition ${\dred \phi}(s_{1})={\dred
  \phi}(s_{2})$ to Def.~\ref{def:sempol}. In general, we can consider a different declassification policy $\phi_{j}$ for each observation point $o_{j}$, $j\in\tO$.
\begin{definition}[Trace-based $\DNI$.]\label{def:propol}
Let $\tO$ be a set of observation points, $\pi_\tI$ a partition of inputs and $\pi_\tO$ the partitions of observable outputs. Let $\forall j\in\tO$ $\phi_{j}$ be the input property declassified in $o_{j}:
\forall j\in\tO.\:\forall s_1,s_2\in\Sigma_{\vdash}\:.\:
{\dred \phi_{j}}(s_{1})={\dred \phi_{j}}(s_{2})\Ra\ 
\post^{j}_{P}(s_1)^\obs= \post^{j}_{P}(s_2)^\obs$
\end{definition}
This definition allows a versatile interpretation of the relation between the {\em where} and the {\em what} dimensions for declassification \cite{BM08}. Indeed if 
we have a unique declassification policy, $\phi$, that holds for all the observation points, then it means that $\forall j\in\tO.\phi_{j}=\phi$. On the other hand, if we have explicit declassification, then we have two different choices. We can combine {\em where} and {\em what} supposing that the attacker knowledge can only increase. 
We call this kind of declassification  {\em incremental} declassification. On the other hand, we can combine {\em where} and {\em what} in a stricter way. Namely, the information is declassified only in the particular observation point where it is explicitly declared, but not in the following points. In this case it is sufficient to consider as $\phi_{j}$ exactly the property corresponding to the explicit declassification and we call it {\em localized} declassification. 
This kind of declassification can be useful when we consider the case where the information obtained by the different observation points cannot be combined, for example if different points are observed by different observers, as it can happen in the security protocol context.
%
\paragraph{\em Protection dimension in the completeness formalization.}
Finally, we can model the protection dimension as a completeness problem.
Consider an abstract domain $\phi=\tuple{\phi_\rel,\phi_\obs}\in\uco(\wp(\inp))$, where $\phi_\rel\in\uco(\wp(\relin))$ and $\phi_\obs\in\uco(\wp(\obsin))$. In this case we consider the concrete semantics $\grass{P}$ of the program, since the property $\phi$ is only used for deciding/selecting when we have to check whether non-interference is satisfied or not. 
For this reason, the property $\phi$ is embedded in the input abstract domain in a way similar to what we have done for the output observation \cite{BM08}.\\ 
Consider any $o \in X^\obs$. Define the set
$H_o\defi\{r\in\relin~|~\tuple{r,o}\in X\}$; i.e., given a value $o$,
$H_o$ contains all the relevant values associated with $o$ in $X$. Then the selecting abstract domain, $\prabs{\cH}{\phi}(X)$, corresponding to $X$,
is defined as $ \prabs{\cH}{\phi}(X) = \bigcup_{o\in X^\obs}
\phi_\rel(H_o)\times\phi_\obs(o) $. Note that the domain $\cH$, for standard $\NI$, is the instantiation of $\prabs{\cH}{\phi}$, where $\phi_\rel$ maps
any set to $\top$ and $\phi_\obs$ maps any set to itself. 

Let $\phi\in\uco(\wp(\inp))$, selecting the inputs on which to check non-interference, then we define the following completeness equation
\begin{eqnarray}
\cH\comp\grass{P}\comp\prabs{\cH}{\phi} = \cH\comp\grass{P}
\end{eqnarray}
\noindent
Now we can connect
$\prabs{\cH}{\phi}$ to $\DNI$:
\begin{theorem}\label{th:decla}
$P$ satisfies
$\DNI$ w.r.t.\ $\phi$ iff
$\cH\comp\grass{P}\comp\prabs{\cH}{\phi}=\cH\comp\grass{P}$.
\end{theorem}
\subsection{All together...}\label{sec:ani}
The following definition combines all the three dimensions obtaining the notion of $\ANI$ as a generalization
of the standard one \cite{GMproof10}. We can say that the idea of abstract non-interference is that a
program $P$ satisfies abstract non-interference relatively to a pair
of observations $\Oin$ and $\Oout$, and to a property $\Sel$, denoted
$\gani{\Sel}{\Oin}{P}{\Oout}$, if, whenever the input values have the
same property $\Sel$ then the best correct approximation of the
semantics of $P$, w.r.t.\ $\Oin$ in input and $\Oout$ in output, does not
change. This captures precisely the intuition that
$\Sel$-indistinguishable input values provide
$\Oin,\Oout$-indistinguishable results, for this reason it can still
be considered a non-interference policy.
\begin{definition}[(Declassified) Abstract
    non-interference $\DANI$]\label{def-abs-sec}\ \\
Let $\Sel,\Oin\in\uco(\wp(\inp)),\Oout\in\uco(\wp(\obsout))$.
\begin{center}
\framebox{
\begin{tabular}{c}
A program $P$ satisfies $\gani{\Sel}{\Oin}{P}{\Oout}$ w.r.t.\ $\pi_\tI$ and $\pi_\tO$ if\\ 
$\forall
x_1,x_2\in\inp\:.\:
{\dred \Sel}(x_1)={\dred \Sel}(x_2)\ \Ra\ {\dred \Oout}((\grass{P}({\dred \Oin}(x_1))))=
{\dred \Oout}((\grass{P}({\dred \Oin}(x_2))))$
\end{tabular}}
\end{center}
\end{definition}
\vspace{.5cm}
For instance, in Eq.~\ref{defNI} we have $\Sel=\tuple{\topi,\ido}$,
$\Oin=\id$ and $\Oout=\id$.
%
In the following, we define closures on $\val^{n}$ by using closures on $\val$. In this case we abuse notation by supposing that $\rho(\tuple{x,y})=\tuple{\rho(x),\rho(y)}$.
\begin{example}\label{definite}
  Consider the property $\Sign$ and $\Par$ defined in Sect.~\ref{intro}.
  Consider $\Pobsin=\{x\}$, $\Prelin=\{y\}$ and
  $\inp=\nint$. Let $\Sel=\Sign$, $\Oin=\id$, $\Oout=\Par$, and consider the program fragment:
\[
P~\defi~ x:=2*x*y^2;
\]

\noindent
In the standard notion of non-interference
there is a flow of information from variable $y$ to variable $x$,
since $x$ depends on the value of $y$, i.e., the statement does not
satisfy non-interference.\\ Let us consider
$\gani{\Sign}{\id}{P}{\Par}$. If 
$\Sign(\tuple{x,y})=\tuple{\Sign(x),\Sign(y)}=\tuple{0+,0+}$, then the possible outputs are always in
$\ev$, indeed the result is always even because there is a
multiplication by $2$. The same holds if $\Sign(\tuple{x,y})=\tuple{0-,0-}$.
Therefore any possible output value, with a fixed observable input,
has the same observable abstraction in $\Par$, which is
$\ev$. Hence $\gani{\Sign}{\id}{P}{\Par}$ holds.
\end{example}
\paragraph{\em The completeness formalization.}
All the completeness characterizations provided embed the abstraction in a different position inside the completeness equation. This makes particularly easy to combine all the completeness characterizations in order to obtain the completeness formalization of $\DANI$.
\begin{theorem}
$P$ satisfies
$\DANI$ w.r.t.\ $\phi,\eta\in\uco(\wp(\inp))$ and $\rho\in\uco(\wp(\oup))$ iff
\[
\begin{array}{ccc}
\mbox{Static attack}:&\qquad &\mbox{Dynamic attack}:\\
\puabs{\cH}{}\comp\pruabs{\grass{P}}{\eta}{\rho}\comp\prabs{\cH}{\phi}=\puabs{\cH}{}\comp\pruabs{\grass{P}}{\eta}{\rho}& &\puabs{\cH}{\rho}\comp\puabs{\grass{P}}{}\comp\prabs{\cH}{\phi}=\puabs{\cH}{\rho}\comp\puabs{\grass{P}}{}
\end{array}
\]
\end{theorem}

\section{Abstract Non-Interference in Language-based Security}\label{ANIsec}

In the previous sections we introduced non-interference as a generic notion that can be used/applied in many fields of computer science, i.e., wherever we have to analyze a dependency relation between two sets of variables.
However, this notion was introduced in the field of language-based security \cite{cohen77}.
In this context we consider the same partition of input and output data into two types: private/confidential (the set of values for variables of type $\tH$ is denoted $\val^{\tH}$) and public (the set of values for variables $\tL$ is denoted $\val^{\tL}$). In this case, non-interference requires that by observing the public output, a malicious attacker must not be able to disclose any information about the private input. It is straightforward to note that this is an immediate instantiation of the general notion we provided with attribute independent abstractions, where $\obsin=\obsout=\val^{\tL}$ and $\relin=\relout=\val^{\tH}$. 

In the context of language-based security the  limitation of the standard notion of non-interference is even more relevant, in particular we can observe
that, in general, it results in a extremely restrictive policy. Indeed, non-interference
policies require that any change upon confidential data must not be
revealed through the observation of public data. There are at least
two problems with this approach. On one side, many real systems are
intended to leak some kind of information. On the other side, even if
a system satisfies non-interference, static checking being approximate
could reject it as insecure. Both of these observations address the
problem of {\em weakening\/} the notion of non-interference for both
characterizing the information that {\em is allowed\/} to flow, and
modeling {\em weaker\/} attackers that can observe only some
properties of public data.  Abstract non-interference can provide formal tools for reasoning about both these kind of weakenings of non-interference in language-based security.
There exists a large literature on the
problem of weakening non-interference (see \cite{M05Aplas} and \cite{BM08} for a formalization of the relation between abstract non-interference and existing approaches).
In the following, we show how we can choose and combine the different policies obtaining the different notions of abstract non-interference introduced in the last years for language-based security.

\subsection{{\em Who} is attacking?}
Let us focus first on the attacker. By attacker we mean the agent aiming at learning some confidential information by {\em analyzing}, to the best of its possibilities, the system under attack. 
In general the attacker is characterized by the public observation of both input and output \cite{GM04popl}.
Depending on the dimension (observation or protection) we use for modeling the {\em input} attacker observation we obtain different notions of abstract non-interference. In particular, if the attacker {\em selects} the computations whose inputs have a fixed public property,  we obtain the so called {\em narrow} non-interference \cite{GM04popl}. If, instead the attacker performs a {\em static} analysis of the code with a generic input observation that may not be the same as the output observation, then we obtain the so called (strictly) abstract approach. In both cases the attacker is characterised by means of two abstractions: the abstract observation of the public {\em input} $\delta$ and the abstract observation of the  public {\em output} $\rho$, both modelled as abstract domains.
\subsubsection{The {\em narrow} approach to non-interference}
In \cite{GM04popl,GMadj10} the notion of {\em narrow (abstract)
non-interference} ($\NANI$ for short) represents a first weakening of standard
non-interference relative to a given model of an attacker. 
The idea behind this notion is to consider attackers that can only dynamically analyze the program by analysing properties of the collected executions. In particular the input public observation {\em selects} the computations that have to satisfy non-interference, while the public output observation represents {\em dynamic observational} power of the attacker.
Formally, given $\delta=\phi_\tL,\:\rho\in\uco(\wp(\val^{\tL}))$, 
respectively, the input and the output observation, we say that a  
program $P$ satisfies narrow non-interference ($\NANI$), written
$\dsecr{\phi_\tL}{P}{\rho}$, if 
\begin{center}
\framebox{$
\forall x_1,x_2\in\val\:.\:
{\dred \phi_\tL}(x_1^{\tL})={\dred \phi_\tL}(x_2^{\tL})\ \Ra\ 
{\dred \rho}((\grass{P}(x_1))^\tL)= {\dred \rho}((\grass{P}(x_2))^\tL)
$}
\end{center}
This notion is a particular instantiation of $\DANI$ where $\eta=\id$ and $\phi=\tuple{\id_\tH,\phi_\tL}$ (where $\id_\tH\defi\lambda x\in\val^\tH.\:x$ and $\id$ is defined on $\val$).
This notion corresponds to other weakenings of non-interference existing in the literature (e.g., PERs \cite{SS-HOSC01}).
\subsubsection{The {\em abstract} approach to non-interference}
A different abstract interpretation-based approach to non-interference can be obtained by modelling attackers as static analyzers of programs. In this case, we check non-interference by considering the best correct approximation of the program semantics in the abstract domains modelling the attacker.
Formally, the idea is to compute the semantics on abstract values, obtaining again a notion of non-interference where only the private input can vary. 
What  we obtain is a policy such
that when the attacker is able to {\em observe} the property $\delta=\eta_\tL$ of
public input, and the property $\rho$ of public output, then no
information flow concerning the private input
is detected by observing the public output.  We call this
notion {\em abstract non-interference\/} ($\ANI$ for short).  A program $P$
satisfies abstract non-interference, written
$\asecr{\eta_\tL}{P}{\rho}$, if
\begin{center}
\framebox{$
\forall x_1,x_2\in\val\:.\: x_1^{\tL}=x_2^{\tL}\ \Ra\ 
\rho(\grass{P}(x_1^{\tH},{\dred \eta_\tL}(x_{1}^{\tL}))^\tL)= \rho(\grass{P}(x_2^{\tH},{\dred \eta_\tL}(x_{2}^{\tL}))^\tL)
$}\end{center}
where we abuse notation denoting by $\grass{P}$ the additive lift, to sets of states, of the denotational semantics of $P$. This notion is an instantiation of $\DANI$ where $\phi=\id$ and $\eta=\tuple{\id_\tH,\eta_\tL}$.
Note that $\asecr{\id}{P}{\id}$ (equivalent to $\dsecr{\id}{P}{\id}$) models exactly (standard)
non-interference. 
\begin{proposition}
The notion of $\ANI$ $\asecr{\delta}{P}{\rho}$ defined above is equivalent to the standard notion of $\ANI$ \cite{GM04popl}, i.e.,
$
\forall x_1,x_2\in\val\:.\: \delta(x_1^{\tL})=\delta(x_2^{\tL})\ \Ra\ 
\rho(\grass{P}(x_1^{\tH},\delta(x_{1}^{\tL}))^\tL)= \rho(\grass{P}(x_2^{\tH},\delta(x_{2}^{\tL}))^\tL)
$.
\end{proposition}
\begin{proof}
We have to prove that $\forall x_1,x_2\in\val$ we have that $(1) x_1^{\tL}=x_2^{\tL}\ \Ra\ \rho(\grass{P}(x_1^{\tH},\delta(x_{1}^{\tL}))^\tL)= \rho(\grass{P}(x_2^{\tH},\delta(x_{2}^{\tL}))^\tL)$ iff  $(2) \delta(x_1^{\tL})=\delta(x_2^{\tL})\ \Ra\ 
\rho(\grass{P}(x_1^{\tH},\delta(x_{1}^{\tL}))^\tL)= \rho(\grass{P}(x_2^{\tH},\delta(x_{2}^{\tL}))^\tL)$. Suppose $(1)$ holds, and consider $x_1^\tL\neq x_2^\tL$ such that $\delta(x_1^{\tL})=\delta(x_2^{\tL})$. From $(1)$ we have that $\rho(\grass{P}(x_1^{\tH},\delta(x_{1}^{\tL}))^\tL)= \rho(\grass{P}(x_2^{\tH},\delta(x_{1}^{\tL}))^\tL)$, but since $\delta(x_1^{\tL})=\delta(x_2^{\tL})$ we have also $\rho(\grass{P}(x_1^{\tH},\delta(x_{1}^{\tL}))^\tL)= \rho(\grass{P}(x_2^{\tH},\delta(x_{2}^{\tL}))^\tL)$. Suppose $(2)$ holds and that $x_1^{\tL}=x_2^{\tL}$, then trivially $\delta(x_1^{\tL})=\delta(x_2^{\tL})$ and by $(2)$ we have $\rho(\grass{P}(x_1^{\tH},\delta(x_{1}^{\tL}))^\tL)= \rho(\grass{P}(x_2^{\tH},\delta(x_{2}^{\tL}))^\tL)$.
\end{proof}

\subsection{{\em What} the attacker may disclose?}

In this section, we focus on another important aspect: what is released or protected. 
Again, depending on which dimension (observation or protection) we use for modeling the information to protect we obtain different approaches to declassification \cite{GM04popl,BM08}. In particular, if the observation dimension is specified, then we are fixing the property that must not be released, namely we model the property whose variation must not be observable, obtaining the so called {\em block} approach to non-interference. On the other hand, if we use the protection dimension, then we are fixing the property that may be released, namely the property of private inputs such that non-interference has to be satisfied for all the private inputs with the same property. In other words, the property whose variation {\em may} interfere. In this case we obtain the so called {\em allow} approach to non-interference \cite{GM04popl}. Note that in this context, we consider declassification while ignoring {\em where} \cite{SSJCS07} declassification takes place.
\subsubsection{The {\em block} approach to declassification}
Let us describe the {\em block approach} introduced in the original notion of abstract non-interference 
\cite{GM04popl}. Note that in standard non-interference we have to protect the {\em value} of private data. If we interpret this fact from the point of view of what we have to keep secret, then we can say that we want to block the {\em identity} property of the private data domain. In the definition, we make the private input range in the domain of values and we check if these
changes are detectable from the observation of the public output.
Suppose, for instance, that we are interested in keeping secret the parity of input private data. Then we make the private input range over the abstract domain of parity, so we check if there is a variation in the public output only when the parity of the private input changes\footnote{This idea was partially introduced  in the use of {\em abstract variables} in the semantic approach of Joshi and Leino \cite{joshiLeino00}, where the variables represent set of values instead of single values. But in \cite{joshiLeino00} only an example of these ideas is provided.}.
Hence, the fact that we want to protect parity is modelled by observing that the {\em distinction} between even and odd private inputs corresponds exactly to what must not be visible to a public output observer. Formally, consider an abstract domain $\eta_\tH\in\uco(\val^{\tH})$ modelling the private information we want to keep secret. A program $P$ satisfies non-interference declassified via blocking ($\BDNI$ for short), written $\gbsecr{\id}{P}{\eta_\tH}{\id}$, if
\begin{center}
\framebox{$ 
\forall x_1,x_2\in\val\:.\:x_1^{\tL}=x_2^{\tL}\ \Ra\ 
\grass{P}({\dred \eta_\tH}(x_1^{\tH}),x_{1}^{\tL})^\tL= \grass{P}({\dred \eta_\tH}(x_2^{\tH}),x_{2}^{\tL})^\tL
$}\end{center}
This is a particular instantiation of $\DANI$ where $\phi=\id$ and $\eta=\tuple{\eta_\tH,\id_\tL}$.
We can obtain the narrow and the abstract declassified forms simply by combining this notion with each one of the previous notions \cite{GM04popl}. In this case it is the semantics that {\em blocks} the flow of information, exactly as the square operation hides the input sign of an integer.
%
%
%
\subsubsection{The {\em allow} approach to declassification}
Finally, let us introduce the {\em allow approach} to declassification. This is a well-known approach, which has been introduced and enforced in several ways in the literature \cite{M05Aplas,SSJCS07}. The idea is to fix which aspects of the private information can be observed by an unclassified observer. 
From this point of view, the standard notion of non-interference, where nothing has to be observed, can be interpreted by saying that only the property $\topi$ (i.e., ``I don't know the private property'') is declassified. This corresponds to saying that we have to check non-interference only for those private inputs that have the same declassified property, noting that all the values are mapped to $\top$ and hence have the same abstract property.
Suppose, for instance,  that we want to downgrade the parity. This means that we do not care if the observer sees any change due to the variation of this private input property. For this reason, we only check the variations of the output when the private inputs have the same parity property. Formally, consider an abstract domain $\phi_\tH$ modelling the private information that is downgraded. A program $P$ satisfies non-interference declassified via allowing ($\ADNI$ for short), written $\gasecr{\id}{P}{\phi_\tH}{\id}$, if
\begin{center}
\framebox{$ 
\forall x_1,x_2\in\val\:.\:x_1^{\tL}=x_2^{\tL}\ \wedge\ {\dred \phi_\tH}(x_{1}^{\tH})={\dred \phi_\tH}(x_{2}^{\tH})\ \Ra\ 
\grass{P}(x_1^{\tH},x_{1}^{\tL})^\tL= \grass{P}(x_2^{\tH},x_{2}^{\tL})^\tL
$}\end{center}
Also this notion is an instantiation of $\DANI$ where $\eta=\id$ and $\phi=\tuple{\phi_\tH,\id_\tL}$.
This is a weakened form of standard $\NI$ where the test cases are reduced, allowing some confidential information to flow. \\

\noindent
In order to make clear the differences and the analogies between the several notions introduced, we collect together, in the next table, all these notions.
\begin{center}
\begin{tabular}{|l|c|c|}
\hline
{\sc Observation dimension} & {\sc Declassification dimension}&{\sc Property}\\
\hline
{\sl Input:} Dynamic protection of inputs ($\phi_\tL$)  & No Declassification &  $\NANI$\\
{\sl Output:} Dynamic analysis of computations ($\rho$)&&$\dsecr{\phi_\tL}{P}{\rho}$\\
\hline
Static analysis of code:&No Declassification&$\ANI$\\
$\eta_\tL$ in input, $\rho$ in output&&$\asecr{\eta_\tL}{P}{\rho}$\\
\hline
No abstraction & $\eta_\tH$ to protect &$\BDNI$\\
&&$\gsecr{\id}{P}{\eta_\tH}{\id}$\\
\hline
No abstraction & $\phi_\tH$ declassified &$\ADNI$\\
&&$\gnsecr{\id}{P}{\phi_\tH}{\id}$\\
\hline\end{tabular}
\end{center}

\subsection{{\em Where} the attacker may observe?}\label{where}
In this section, we show that we can exploit the semantic dimension in order to model also {\em where} the attacker may observe public data. In particular, we need to consider trace-based $\NI$ for fixing also the observable program points, where the attacker in some way can access the intermediate observable computations.
Let us recall trace-based $\NI$ for language-based security
\[
\begin{array}{c}
\forall j\in\tO.\:\forall s_{1},s_{2}\in \Sigma_{\vdash}.\:
 s_{1}^\tL=s_{2}^\tL\ \Ra\ \post^{j}_{P}(s_{1})^\tL=\post^{j}_{P}(s_{2})^\tL
\end{array}
\]
and its completeness characterization in terms of weakest precondition semantics:
\[
\cH\comp\widetilde{\pre}^{j}\comp\cH=\widetilde{\pre}^{j}\comp\cH.
\] 

These results say that, in order to check $\DNI$ on traces, we
would have to make an analysis for each observable program point $j$, combining,
afterwards, the information disclosed. In order to make only one
iteration on the program even when dealing with traces, our basic idea
is to combine the weakest precondition semantics, computed at each
observable point of the execution, together with the observation of
public data made at the particular observation point.

We will describe our approach on our running example,
Ex.~\ref{ex:infrelds}, where $o_{i}$ and $p_{i}$ denote
respectively the observation and the protection points of the program
$P$.
\begin{example}\label{ex:infrelds}
\[
P=
\left [
\begin{array}{lll}
{\dred p_{0}}\ra&\qquad \{{\dred h_{2}\modu 2=a}\}&\la {\dgreen o_{0}}\\
&h_{1}:=h_{2};&\\
&h_{2}:=h_{2}\modu 2;&\\
&l_{1}:=h_{2};&\\
&h_{2}:=h_{1}&\\
&l_{2}:=h_{2};&\\
&l_{2}:=l_{1}&\\
&\qquad \{{\dgreen l_{1}=l_{2}=a}\}&\la {\dgreen o_{6}}
\end{array}
\right .
\]
 In this example, the observation in
output of $l_{1}$ or $l_{2}$ allows us to derive the information
about the parity of the secret input $h_{2}$.
\end{example}
%
%
In the presence of explicit declassifications, it allows a precise characterization of the  relation between the {\em where} and the {\em what} dimensions of non-interference policies. The idea is to track (by using the wlp computation) the information disclosed in each observation point till the beginning of the computation, and then to compare it with the corresponding declassification policy.
%
%
The next example shows how we track the information disclosed in each observable program point.
We consider as the set of observable program points, $\tO$, the same set used for gradual release, namely, the program points corresponding to low events.  However, in general, our technique allows $\tO$ to be any set of program points. When there is more than one observation point, we will use the notation $[\Phi]^{O}$ to denote that the information described by the assertion $\Phi$ can be derived from the set of observation points in $O$.

\begin{example}\label{ex:infreltr}
\[
P=
\left [
\begin{array}{lll}
{\dred p_{0}}\ra&\qquad \{{\dred  [h_{2}=b]^{o_{5}},[h_{2}\modu 2=a]^{o_{3},o_{5},o_{6}}},{\dblue [l_{2}=c]^{o_{3},o_{0}},}{\dgreen [l_{1}=d]^{o_{0}}}\}&\la {\dgreen o_{0}}\\
&h_{1}:=h_{2};&\\
&\qquad \{{\dblue [h_{1}=b]^{o_{5}},[h_{2}\modu 2=a]^{o_{3},o_{5},o_{6}},[l_{2}=c]^{o_{3}}}\}&\\
&h_{2}:=h_{2}\modu 2;&\\
&\qquad\{ {\dblue [h_{1}=b]^{o_{5}},[h_{2}=a]^{o_{3},o_{5},o_{6}},[l_{2}=c]^{o_{3}}}\}&\\
&l_{1}:=h_{2};&\\
&\qquad \{{\dblue  [h_{1}=b]^{o_{5}},[l_{1}=a]^{o_{3},o_{5},o_{6}}},{\dgreen [l_{2}=c]^{o_{3}}}\}&\la {\dgreen o_{3}}\\
&h_{2}:=h_{1}&\\
&\qquad \{{\dblue [h_{2}=b]^{o_{5}},[l_{1}=a]^{o_{5},o_{6}}}\}&\\
&l_{2}:=h_{2};&\\
&\qquad \{{\dblue [l_{1}=a]^{o_{5},o_{6}}},{\dgreen [l_{2}=b]^{o_{5}}}\}&\la {\dgreen o_{5}}\\
&l_{2}:=l_{1}&\\
&\qquad \{{\dgreen l_{1}=l_{2}=a}\}&\la {\dgreen o_{6}}
\end{array}
\right .
\]
\end{example}
For instance, at observation point $o_5$, $[l_1 = a]^{o_{5},o_{6}}$ is obtained either via
wlp calculation of $l_2 := l_1$ from $o_{6}$, or via the direct observation of public data in $o_{5}$, while $[l_2 = b]^{o_{5}}$ --- where $b$ is an arbitrary symbolic value --- is only due to the observation of the value $l_2$ in $o_5$. The assertion in $o_3$ --- $[h_1 = b]^{o_{5}}, [l_1 = a]^{o_{3},o_{5},o_{6}}$ --- is obtained by computing the wlp semantics $\Wlp(h_2 := h_1, ([h_2 = b]^{o_{5}}, [l_1 = a]^{o_{3},o_{5},o_{6}}))$, while $ [l_{2}=c]^{o_{3}}$ is the observation in $o_3$. Similarly, we can derive all the other assertions.

It is worth noting, that the attacker is more powerful than the one considered in 
Example~\ref{ex:infrelds}. In fact, in this case, the possibility of observing $l_{2}$ in the program point $o_{5}$ allows to derive the exact (symbolic) value of $h_{2}$ ($h_{2}=b$ where $b$ is the value observed in $o_{5}$).  This was not
possible in Example~\ref{ex:infrelds} based on simple input-output,
since the value of $l_{2}$ was lost in the last assignment.

Now, we can use the information disclosed in each observation point, characterized by computing the wlp semantics, in order to check if the corresponding declassification policy is satisfied or not. This is obtained simply by comparing the abstraction modelling the declassification with the state abstraction
corresponding to the information released in the lattice of abstract interpretations. In the sequel, we will consider a slight extension of a standard imperative language with explicit syntax for declassification as in Jif~\cite{jif} or in the work on
delimited release~\cite{SM04}. Such syntax often takes the form
$l:=\declassify(e(h))$, where $l$ is a public variable, $h$ is 
a secret variable and $e$ is an expression containing the variable $h$. This syntax means that the final value of $e(h)$, corresponding to some information about $h$, can safely made public assigning it to a public variable $l$.
\begin{example}
\label{ex:expdec}
\comment{ \[
P=
\left [
\begin{array}{lll}
{\dred p_{0}}\ra&\qquad \{{\dred  h_{2}=b}, {\dblue [h_{2}\modu 2=a]^{o_{3}}, l_{1}=d,l_{2}=c}\}
&\la {\dgreen o_{0}}\\
&h_{1}:=h_{2};&\\
&\qquad \{{\dblue h_{1}=b, [h_{2}\modu 2=a]^{o_{3}}}\}&\\
&h_{2}:=h_{2}\modu 2;&\\
&\qquad\{ {\dblue h_{1}=b,l_{2}=c},{\dblue [h_{2}=a]^{o_{3}}}\}&\\
&l_{1}:=\declassify^{o_{3}}(h_{2});&\\
&\qquad \{{\dblue  h_{1}=b,l_{1}=a}\},\{{\dgreen l_{2}=c}\}&\la {\dgreen o_{3}}\\
&h_{2}:=h_{1}&\\
&\qquad \{{\dblue h_{2}=b,l_{1}=a}\}&\\
&l_{2}:=h_{2};&\\
&\qquad \{{\dblue l_{1}=a}\},\{{\dgreen l_{2}=b}\}&\la {\dgreen o_{5}}\\
&l_{2}:=l_{1}&\\
&\qquad \{{\dgreen l_{1}=l_{2}=a}\}&\la {\dgreen o_{6}}
\end{array}
\right .
\]}
Consider the program in Example~\ref{ex:infreltr}, with the only difference that in $o_{3}$ the statement $l_{1}:=h_{2}$ is substituted with $l_{1}:=\declassify(h_{2})$. In this case, the corresponding declassification policy is $\phi_{3}=\id_{h_{2}}$. Now we can compare this policy with the private information disclosed in $o_{3}$, which is  $h_2 \modu 2 = a$ ($h_{2}$'s parity) and therefore we conclude that the declassification policy in $o_{3}$ is satisfied because parity is more abstract than identity.
Nevertheless, the program releases the information $h_2=b$ in the program point  $o_{5}$. This means that the security of the programs depends on the kind of localized declassification we consider. For incremental declassification the program is secure since the release is licensed by the previous declassification since the observation point $o_{5}$, where we release information corresponds to the declassification policy $\phi_{5}=\phi_{3}=\id_{h_{2}}$, while for localized declassification the program is insecure since there isn't a corresponding 
declassification policy at $o_5$ that licenses the release. 
\end{example}

\subsection{Certifying abstract non-interference}\label{cert}
In this section we briefly recall the main techniques proposed for certifying different aspects of $\ANI$ in language-based security.

\paragraph{Characterizing attackers.}
In \cite{GM04popl}, a
method for deriving the most concrete output observation for a
program, given the input one, is provided.  In particular, we have interference when the attacker observes too much, namely when it can {\em distinguish} elements due to different confidential inputs. The idea is to collect together all such observations.
We don't mean here to enter in the details of this characterization (see \cite{GM04popl} and \cite{GMadj10}), we simply describe the different steps we perform in order to characterize the most concrete harmless output (dynamic) observer.
By dynamic we mean that we can only characterize the output observation $\rho$ but we cannot characterize the input one, $\eta$, which remain fixed.
\begin{enumerate}
\item First of all we define the sets of elements that, depending on the fixed (abstract) input, must not be distinguished. In other words, if the inputs have to agree simply on the public part, then we collect in the same set all the outputs due to the variation of the confidential inputs. For instance, for the property $\ANI$ $\asecr{\delta}{P}{\rho}$ we are interested in the elements of 
\[
\undist{\ANI(\delta)}{P}\defi\sset{\sset{\grass{P}(h,\delta(l))}{h\in\val^\tH}}{l\in\val^\tL}
\]
Namely, each $Y\in\undist{\ANI(\delta)}{P}$ is a set of elements indistinguishable for guaranteeing non-interference.
\item Let $\Pi$ the property to analyze ($\ANI$ for instance), the maximal harmless observation is the set of all the elements satisfying the predicate $\Secr{P}{\Pi}$:
\[
\forall X\in\wp(\val^{\tL})\:.\:\Secr{P}{\Pi}(X)\ \Lra\ (\exists Z\in\undist{\Pi}{P}.Z\subseteq X\ \Ra\ \forall W\in\undist{\Pi}{P}.W\subseteq X)
\]
Namely the maximal subset of the power set of program states such that any of its elements can only completely contain sets in $\undist{\Pi}{P}$, namely does not break indistinguishable sets.
\end{enumerate}

\paragraph*{Characterizing information released.}\label{infoReleased}
Abstract non-interference can be exploited also for characterizing  the maximal amount of confidential information released, modelled in terms of the most concrete confidential property that can be analyzed by a given attacker \cite{GM04popl}.
 The idea is to find
the maximum amount of information disclosed by computing
the most abstract property on confidential data which has to be
declassified in order to guarantee secrecy. In other words, we characterize  the minimal aggregations of confidential inputs that do not generate any variation in the public output. 
\begin{enumerate}
\item We collect together all  and only the confidential values that do not generate a variation in the output observation, for instance for $\ANI$ we consider
\[
\Pi_\tP(\delta,\rho)\defi
\{\tuple{\{h\in\val^\tH~|~\rho(\fP(\tuple{h,\delta(l)})^\tL)=
A\},\delta(l)}~|~l\in\val^\tL,\ A\in\rho\}
\]
which, for each public input $l$ and for each output observation $A$ gathers all the confidential inputs generating precisely the observation $A$.
\item We use this information for generating the abstraction $\Phi$, of confidential inputs such that any pair of inputs in the same set does not generate a flow, while any pair of inputs belonging to different sets do generate a flow:
\[
\Pi_\tP(\delta,\rho)_{|L}\defi\{H~|~\tuple{H,L}\in\Pi_\tP(\delta,\rho)\}\qquad \Phi\defi\ok{\cP(\bigsqcap_{L\in\delta} \cM(\Pi_\tP(\delta,\rho)_{|L}))}
\]
This domain is the most abstract property that contains all the
possible variations of confidential inputs that generate insecure
information flows, and  it is the most concrete such that each variation
generates a flow. It uniquely represents the confidential
information that flows into the public output \cite{GM04popl}.
\end{enumerate}

\paragraph*{A proof system for Abstract non-interference} In the previous sections, we start from the program and in some way we characterize the abstract non-interference policy that can be satisfied by adding some weakenings, in the attacker model or in the information allowed to flow. In the literature, we can find another certification approach to abstract non-interference. In particular, in \cite{GMproof10} the authors provide a {\em proof system} for certifying $\ANI$ in programming languages. In this way, they can prove, inductively on the syntactic structure of programs, whether a specific $\ANI$ policy is or not satisfied.

\section{A promising approach in other security fields}\label{othSec}
In this section, we provide the informal intuition of how interference properties pervade two challenging IT security fields and therefore of how approaches based on abstract non-interference may be promising into these fields and deserve further research.
\paragraph*{Code injection.}
Among the major security threats in the web application context we have {\em code injection attacks}. Typically, these attacks inject some inputs, interpreted as executable code, which is not adequately checked and which {\em interfere} with the surrounding code in order to disclose information or to take the control of the system. Examples of these kinds of attacks are: SQL injections, Cross-site scripting (XSS), Command injections, etc. In order to protect a program from these kind of attacks we first can verify whether the code coming from potentially untrusted sources may interfere with the sensible components of the system.  The idea of tainted analysis is that of following the potentially untrusted inputs (tainted) checking whether they get in contact with sensible containers or data. Instead, the idea based on (abstract) non-interference \cite{MB10} uses the weakest precondition approach due to completeness characterization of non-interference \cite{GMadj10,BM08} introduced in Sect.~\ref{where} in order to characterize which variables, i.e., information, the injected code may manipulate in order to guarantee the absence of insecure information flows. The idea in this case is to take into account also the implicit protection that the code itself may provide. In \cite{MB10} this idea is precisely formalized and here we describe it by using a XSS attack example.
Suppose a user visits an untrusted web site in order to download a picture, where an attacker has inserted his own malicious Javascript code (Fig.~\ref{XSS}), and execute it on the clients browser \cite{NJKKV07}.
In the following we describe a simplified version. The Javascript code snippet in Fig.~\ref{XSS}  can be used by the attacker to send cookies\footnote{A cookie is a text string stored by a user's web browser. A cookie consists of one or more name-value pairs containing bits of information, sent as an HTTP header by a web server to a web browser (client) and then sent back unchanged by the browser each time it accesses that server. It can be used, for example, for authentication.} of the user to a web server under the attackers control.
\begin{figure}[t]
\begin{flushleft}
\end{flushleft}
\begin{Verbatim}[frame=single,fontsize=\scriptsize]
   var cookie = document.cookie; /*initialisation of the cookie by the server*/
   var dut;
   if (dut == undefined) {dut = "";}
   while(i<cookie.length) {
       switch(cookie[i]) {
         case 'a': dut += 'a'; break;
         case 'b': dut += 'b'; break;
         ...
     }
   }      /* dut contains now copy of cookie*/
   document.images[0].src =  "http://badsite/cookie?" + dut; 
   /* When the user click on the image dut is sent to the web server under the attackers control */
\end{Verbatim}
\caption{Code creating a XSS vulnerability.}\label{XSS}
 \end{figure}
\noindent
By performing the analysis proposed in \cite{MB10} we can show that confidentiality is violated since there is a (implicit) flow of information from private variable \textit{cookie} towards the public variable \textit{dut}.
This corresponds to the sensitive information disclosed when \textit{dut} is initialised to the empty string. We can also formally show that an attacker can exploit \textit{dut} for disclosing other user confidential information.
Suppose, for instance, the attacker to be interested in the \textit{history} object together with its attributes\footnote{The history object allows to navigate through the history of websites that a browser has visited.}.   \vspace{.1cm}\\
\hspace{-.3cm}
\begin{minipage}{.50\textwidth}
An attack could loop over the elements of the \textit{history} object and pass through variable \textit{dut} all the web pages the client has had access to. Consider for example the injection code on the right.
\end{minipage}
\hspace{1cm}
\begin{minipage}{.40\textwidth}
\vspace{-.25cm}
\begin{Verbatim}[frame=single,fontsize=\scriptsize]
     <script language="JavaScript">
     var dut = "";
     for (i=0; i<history.length; i++){
         dut = dut + history.previous;
     }   </script>
\end{Verbatim}
\end{minipage}
\vspace{.1cm}\\
Hence the $\ANI$-based analysis detect a code vulnerability that the attacker can exploit by inserting the code above just before the malicious code (Fig.~\ref{XSS}) in the untrusted web page, getting both \textit{history} and \textit{cookie} through the variable \textit{dut}.
It is worth noting that this approach provides a theoretical model for the existing techniques used in practice for protecting code from XSS attacks \cite{NJKKV07}.

\paragraph*{Code obfuscation.}
Code obfuscation is increasing its relevance in security, providing an effective way for facing the problem of both source code protection and binary protection. 
Code obfuscation~\cite{CTL98} aims at transforming programs in order to make them more difficult to analyze while preserving their functionality.\\
In \cite{JGM12}, the authors introduce a generalized notion of obfuscation which, while preserving denotational semantics, obfuscates, namely transforms, an abstract property of interest, the one to protect (this is a further generalization of the obfuscation defined in \cite{DG09}). Let $P$ be a program, $\grass{P}$ its denotational semantics, and $\grass{P}^\rho$ the abstract semantics to obfuscate, namely $\rho$ is the abstract semantic property to protect, then $\hat{t}$ is an obfuscator w.r.t. $\rho$ if 
\[
\grass{P}=\grass{\hat{t}(P)}\qquad\mbox{and}\qquad\grass{P}^\rho\neq\grass{\hat{t}(P)}^\rho
\]
Based on this notion of obfuscation, in \cite{JGM12} the authors propose to exploit the analysis of code fragments dependencies in order to obfuscate precisely the code {\em relevant} in the computation of the information to protect. This idea is based on the awareness that the
effect of replacing code in programs can only be understood by the notion of dependency \cite{corecalc}.
Hence, the idea is to model the effect of substituting program fragments with {\em equivalent\/}  obfuscated structures by indirectly considering the $\ANI$ framework. This is an indirect application since it goes through the dependency analysis instead of the non-interference analysis, namely it is more related with slicing, which can be weakened in the so called {\em abstract slicing} \cite{MZ08,MN10}, strongly related to $\ANI$\footnote{Program slicing \cite{Weiser81} is a program manipulation technique which
extracts, from programs, statements relevant to a particular
computation.}. \\
The idea is to characterize the {\em dependency\/} between the computation of the values of the different variables $x_i$ and the observable semantics of the program $\tP$. For this reason the notion of {\em stability\/} is defined \cite{JGM12}, which specifies that a program $\tP$ is stable w.r.t.\ the variables $x_i$, when any change of their abstract (computed) values does not change the observable semantics of the program.
%
%


The notion of stability corresponds to $\ANI$, considering $x_i$ as internal, and observing the whole output.
In code obfuscation, we are interested in the negation of this property which is {\em instability}, i.e., there exist abstract values for $x_i$ that cause a variation in the output observation. In other words, a variation of a fixed internal property $\eta_\rel$ induces a variation in the observable output. This means that the code computing the variables $x_i$ inducing instability is the portion of code that it is sufficient to obfuscate in order to modify the observation of the whole program, deceiving any observer. This is a sufficient condition since, due to instability, we guarantee that the output observation changes. Again this is a new approach to code obfuscation that surely deserves further research.
%

\section{Conclusions}
In this paper we provide a survey on the framework of abstract non-interference developed in the last years.
This generalization of language-based non-interference was introduced in 2004 \cite{GM04popl} as a semantic-based weakening of non-interference allowing some information to flow in order to accept programs, such as password checking, that even if considered acceptable where rejected w.r.t. standard non-interference. We have also shown that this new approach to non-interference was indeed strongly related with the main existing approaches to non-interference \cite{M05Aplas}. Since then, we realized that abstract non-interference was a powerful theoretical framework where it was possible to model and understand several aspects of non-interference policies, modeling also declassification and formally relating different non-interference policies, usually independently introduced  \cite{BM08}. This is proved also by the fact that our approach was probably of inspiration in different recent works such as \cite{BNR08}, where declassification policies are modeled in a weakest precondition style, as we've done in \cite{BGM07}, or in \cite{BDG12} where the authors model declassification as a weakening of what the attacker can observe of the input data.\\
Concluding, we believe that abstract non-interference is a really general framework that can provide useful theoretical bases for understanding ``interactions" in different fields of security, as briefly shown in the last section of this paper. For instance, a challenge for the abstract non-interference approach is to extend it in order to cope with other security policies, such as those proposed by McLean \cite{ML96} and or by Mantel \cite{M2000}, which would allow to consider interactions between functions instead of between variables. This kind of properties could be interesting, for instance, in malware detection for analyzing the interference between a malware and the execution environment \cite{DM13id}.\\ 
Finally, we are obviously aware also of existing limitations of this approach. Among all, our approach is still quite far from a practical exploitation: it is powerful for understanding interference phenomena, but the cost of this power is the distance from a real implementation which is our next big challenge to face.

\paragraph*{Acknowledgements.} I would like to thank Dave Schmidt for the time he dedicated to me during my postdoc in Manhattan and for being so inspiring for the research I developed since then.\\ This work was partially supported by the PRIN 2010-2011 project "Security Horizons"

{\small
\bibliographystyle{eptcs}

}
\end{document}